\documentclass[reqno,centertags,12pt]{amsart}
\usepackage[utf8]{inputenc}
\usepackage{times}

\usepackage{amsthm,amssymb}
\usepackage{mathtools}

\usepackage{enumitem}
\usepackage{dsfont}

\usepackage[breaklinks=true,pagebackref=false]{hyperref}

\usepackage[margin={3cm},marginparwidth={2.5cm}]{geometry}

\newcommand{\N}{{\mathbb{N}}}
\newcommand{\R}{{\mathbb{R}}}

\newcommand{\calC}{\mathcal{C}}
\newcommand{\calD}{\mathcal{D}}

\newcommand{\calQ}{\mathcal{Q}}

\newcommand{\ol}{\overline}

\newcommand{\wti}{\widetilde  }

\newcommand{\loc}{\text{\rm{loc}}}
\newcommand{\la}{\langle}
\newcommand{\ra}{\rangle}

\newcommand{\veps}{\varepsilon}
\newcommand{\re}{\mathrm{Re}}
\newcommand{\im}{\mathrm{Im}}

\allowdisplaybreaks
\DeclareMathOperator{\supp}{supp}

\DeclareMathOperator{\distr}{distr}
\DeclareMathOperator{\dist}{dist}

\newcommand{\ind}{\mathds{1}} 

\newtheorem{theorem}{Theorem}[section]

\newtheorem{lemma}[theorem]{Lemma}

\theoremstyle{definition}
\newtheorem{definition}[theorem]{Definition}
\newtheorem{assumption}[theorem]{Assumption}

\newtheorem{remark}[theorem]{Remark}
\newtheorem{remarks}[theorem]{Remarks}

\numberwithin{theorem}{section}
\numberwithin{equation}{section}
\newcounter{smalllist}
\newenvironment{SL}{\begin{list}{{\rm\alph{smalllist})}}{%
\setlength{\topsep}{0mm}\setlength{\parsep}{0mm}\setlength{\itemsep}{0mm}%
\setlength{\labelwidth}{2em}\setlength{\leftmargin}{2em}\usecounter{smalllist}%
}}{\end{list}}

\usepackage{color}

\thanks{\copyright 2022 by the authors. Faithful reproduction of this article,
       in its entirety, by any means is permitted for non-commercial purposes}

\begin{document}

\title[Quantum Systems at The Brink]{Quantum Systems at the Brink: Existence of Bound States, Critical Potentials and Dimensionality}

\author{Dirk Hundertmark}
 
\address{Department of Mathematics, Institute for Analysis, Karlsruhe Institute of Technology, 76128 Karlsruhe, Germany, and Department of Mathematics, Altgeld Hall, University of Illinois at Urbana-Champaign, 1409 W. Green Street, Urbana, IL 61801, USA
 }
 \email{dirk.hundertmark@kit.edu}
    
\author{Michal Jex}
 
\address{Department of Physics, Faculty of Nuclear Sciences and Physical Engineering, Czech Technical University in Prague, B\v rehov\'a 7, 11519 Prague, Czech Republic, and CEREMADE, Dauphine University, Place du Maréchal de Lattre de Tassigny, 75775 Paris Cedex 16, France
}
  \email{michal.jex@fjfi.cvut.cz}  
\author{Markus Lange}
 
\address{ SISSA, Mathematics Area, Via Bonomea 265, 34136 Trieste, Italy
}
 \email{mlange@sissa.it}

\begin{abstract}
One of the crucial properties of a quantum system is the existence of 
bound states. While the existence of eigenvalues below zero, i.e., 
below the essential spectrum, is well understood, the situation of 
zero energy bound states at the edge of the essential spectrum is far 
less understood. 
We present necessary and sufficient conditions for Schr\"odinger operators 
to have a zero energy bound state. Our sharp criteria show that the 
existence and non-existence of zero energy ground states 
depends strongly on the dimension and the asymptotic behavior of the 
potential. There is a spectral phase transition with dimension 
four being critical. 
\end{abstract}

\maketitle
{\hypersetup{linkcolor=black}\tableofcontents}

\section{Introduction}

The existence of bound states plays a crucial role for the properties of quantum systems.  Of special importance is the ground state, i.e., the eigenfunction corresponding to the lowest eigenvalue of the Hamiltonian describing the system. In this paper we consider a Schr\"odinger operator of the form 
\begin{equation}\label{Hamiltonian}
H=-\Delta+V
\end{equation}
on $L^2(\mathbb R^d)$ where $V\in L^1_{loc}(\R^d)$ is a real--valued potential such that the operator $H$ is a 
well--defined self--adjoint realization of the formal differential operator $-\Delta+V$ which is bounded from 
below. Moreover, we need that eigenfunctions of $H$ are continuous. The precise conditions are given in Assumption \ref{assumption} below. 

We are particularly interested in the special case when the ground state energy of the Schr\"odinger operator $H$ is at the threshold of the essential spectrum. By shifting the potential by a constant, one can assume that the essential spectrum of 
$H$ starts at zero. One also often assumes that  the potential 
$V$ decays to zero at infinity such that the essential spectrum 
$\sigma_{\text{ess}}(H)=[0,\infty)$,  see 
Remark \ref{rem:essential spectrum}. 
Under these conditions the zero--energy level is 
at the edge of two regions with very distinct behavior: the point and the continuous spectrum. 
It is well-known that positive eigenvalues embedded in the continuum appear only due to a special combination of oscillations 
and slow decay of the potential. This goes back to \cite{VonNeuWig-paper-1929}, see also \cite{DenKis07,FraSim17,IonJer03,Sim69} and the references therein. Also, with the help of the min-max theorem, the existence and non-existence of eigenvalues below zero is well-understood, see, e.g., \cite{ReeSim4}.

Whether zero is actually a threshold eigenvalue, i.e., an eigenvalue at \emph{the edge of the continuum} is a very difficult problem, in general. Early results on existence or non-existence of zero-energy eigenvalues go back to
 \cite{Agm70, JeKa79, Ken89, Kno78-1, Kno78-2, Lie81, New77,Ram87, Ram88, Sim81}. 
In \cite{JeKa79} the authors studied the behavior of resonances and eigenstates at the zero--energy threshold in $d=3$. Furthermore, 
 based on the remark of a referee they note that resonances  cannot exist in dimensions $d>4$ based on properties of Riesz potential. However their approach is not applicable for $d=4$. 
 For slowly decaying \emph{negative} potentials which, amongst other conditions, obey $V(x)\sim -c|x|^{-\gamma}$ for some $c>0$ and  $0<\gamma<2$ in the limit $|x|\to\infty$, the non-existence of zero energy eigenstates was shown in \cite{DerSki09,FouSki04}, while it was noted in \cite{bol85} that a long range Coulomb part can create zero energy eigenstates, see also \cite{Nak94, Yaf82}. An analysis of eigenstates 
 and resonances at the threshold for the case of certain nonlocal operators appeared in \cite{KaLo20}. 

In \cite{BenYar90} it has been shown that, for Schr\"odinger operators on $L^2(\R^3)$ with spherical 
 symmetric potentials $V\in L^p(\R^3)$ with $p>3/2$ whose positive 
 part satisfies $V_+(x)\le 3/(4|x|^{2})$ 
 for $|x|$ large enough,  zero is not an eigenvalue corresponding 
 to a positive square integrable ground state eigenfunction. 
 This extends to potentials with  
 $V_+\leq |x|^{-2}\left(3/4+\ln^{-1}(|x|)\right)$ near infinity in $\mathbb R^3$,  
 the constants $3/4$ and $1$ are optimal. 
 For similar results see \cite{GriGar07}, which reproved a slightly 
 weaker non-existence result compared to  
 \cite{BenYar90} and  additionally showed that if 
 $V(x)\ge C|x|^{-2}$ for some constant $C>3/4$ and $|x|$ large then zero is an eigenvalue for critical potential, see Definition~\ref{def:VirtualLevel}. 
 Thus a repulsive part can stabilize zero energy bound states of quantum systems.  
 
 Strictly speaking, the paper \cite{BenYar90} deals with continuous potentials on $\R^3$ but they note that the 
 condition  $V\in L^p(\R^3)$ with $p>3/2$ is enough to guarantee continuity of ground states, due to a Harnack inequality for positive eigenfunctions.   
We also note that compactly supported zero--energy eigenfunctions were 
constructed in \cite{KenNad00,KoTa02} for potential $V\in L^p(\R^d)$ 
with $p<d/2$ and compact support. For these potentials, a 
Harnack inequality for the ground state cannot hold. 

In this paper we significantly extend all previous results, in particular the ones of \cite{BenYar90} and 
\cite{GriGar07}, by proving a  \emph{family of sharp criteria} for the existence and non--existence of zero energy ground states at \emph{the edge of the essential spectrum} 
for Schr\"odinger operators in \emph{arbitrary} dimensions. 
In particular, our results apply to 
Schr\"odinger operators with a so--called virtual level at zero energy and they explain when such 
a virtual level is a true ground state or when it is a resonance. 

Our results clearly explain why increasing the dimension makes it easier for a virtual level to be a 
true ground state. In particular, our work explains why dimension $d=4$ is \emph{critical}. 
Dimension four shares some similarity 
with the case of lower dimensions but  \emph{higher order corrections} from our criteria   
are needed to settle this case.  

\smallskip

Our main assumption on the potential $V$ are given by 
\begin{assumption}\label{assumption}
  The potential $V$ is in the local Kato--class $K_{d,loc}(\R^d)$ 
  and the negative part $V_-= \sup(-V,0)$ is relatively 
  form small w.r.t. $-\Delta +V_+$, i.e., there exist 
  $0\le a<1$ and $b\ge 0$ such that 
  \begin{equation}\label{eq:form small}
  	\la \psi, V_-\psi \ra = \|V_-^{1/2}\psi\|^2 
  		\le a(\|\nabla \psi\|^2 +\|\sqrt{V_+}\psi\|^2) + b\|\psi\|^2
  \end{equation}
  for all $\psi\in H^1(\R^d)\cap\calD(\sqrt{V_+})$. Here 
  $\calD(\sqrt{V_+})$ is the domain of the multiplication operator $\sqrt{V_+}$ on $L^2(\R^d)$, 
  also called the form domain of $V_+$ and often written as $\calQ(V_+)$. 
\end{assumption}

Note that what we call \textit{relatively form small} is usually called \textit{relatively form-bounded with relative bound $a < 1$}. 
We will call a potential $W$ infinitesimally form bounded (w.r.t. $-\Delta+V_+$) if for all $a>0$ there exist $b\ge 0$ such that the positive and negative parts 
$W_\pm$ satisfy \eqref{eq:form small} (with $V_-$ replaced by $W_\pm$).    

\begin{remark}
  The (local) Kato--class $K_{d,loc}(\R^d)\subset L^1_{loc}(\R^d)$, whose definition is recalled below, see \eqref{eq:Kato class}, contains most, if not all physically relevant potentials. 
  This assumption is only made to guarantee that all 
  weak local eigenfunctions of $H$ are continuous, 
  see \cite{AizSim82, Simader90, Sim82}. 
  
  One could relax the assumption that $V\in K^d_{loc}(\R^d)$ to $V\in L^1_{loc}(\R^d)$, if some other 
  condition guaranteed that weak local 
  eigenfunctions of $H$ are continuous. In fact, it would be sufficient to have that  
  eigenfunctions are locally bounded and that a ground state of $H$ is bounded 
  away from from zero on compact sets. 
  As will become clear from the proofs, we can allow for severe local singularities. 
  E.g., it is enough to assume that $V$ is in the local Kato--class outside some compact set $K\subset\R^d$. 
\end{remark}

If $V_\pm\in L^1_{loc}$ and \eqref{eq:form small} holds, 
the KLMN theorem shows that there exists a unique 
self-adjoint operator $H$, informally given by the differential operator 
$-\Delta+V$, such that its 
quadratic form, which with a slight abuse of notation we write as 
  \begin{equation}\label{eq:quadratic form H intro}
  	\la \psi, H\psi \ra \coloneqq
  		\la \nabla\psi,\nabla\psi\ra 
  			+ \la \sqrt{V_+}\psi, \sqrt{V_+}\psi\ra 
			- \la \sqrt{V_-}\psi, \sqrt{V_-}\psi\ra \,
  \end{equation}
  is well-defined for $\psi\in \calQ(H)\coloneqq H^1(\R^d)\cap\calQ(V_+)$. 
  Moreover, it is closed and bounded from below on the quadratic form domain $\calQ(H)$. See also the discussion at the 
  beginning of the next section.

\smallskip
To formulate our main results we recall the definition of the 
iterated logarithms $\ln_n$ defined, for natural numbers $n\in\N$, by  
$\ln_1(r)\coloneqq\ln(r)$ for $r>0$ and inductively for $r>e_n$ by 
$\ln_{n+1}(r)\coloneqq\ln(\ln_n(r))$. Here $e_0=0$ and $e_{n+1}= e^{e_n}$. 
Our first main result can be summarized as follows
\begin{theorem}[Absence of a zero energy ground state]\label{thm:intro absence} 
  Assume that the potential $V$ satisfies Assumption \ref{assumption} 
  and $\sigma(H) = [0,\infty)$. If for some $m\in\N_0$ and $R>e_m$ 
  \begin{equation}\label{eq:intro absence}
  V(x)\leq  \frac{d(4-d)}{4|x|^2}	+\frac{1}{|x|^2}\sum_{j=1}^m\prod_{k=1}^j\ln_k^{-1}(|x|)
  \end{equation}
  for all $|x|\ge R$, then zero is not an eigenvalue of the Schr\"odinger operator $H$. 
\end{theorem}
  As usual the empty product is $1$ and the empty sum equals $0$. 
\begin{remark}
	In particular, if $\inf\sigma_{ess}(H)= 0$ then Theorem \ref{thm:intro absence} shows 
	that zero is not an eigenvalue at the edge of the essential spectrum.  
	Theorem \ref{thm:intro existence} below shows the sharpness of condition \eqref{eq:intro absence} 
	on the potential $V$   
	for the absence of an embedded ground state at the edge of the essential spectrum.
\end{remark}
Our second main result shows that critical potentials create    
zero energy ground states if they are not too small 
at infinity.   We call a potential $W\ge 0$ nontrivial, 
if it is strictly positive on a set of positive 
Lebesgue measure. 

\begin{definition}[Critical potential]\label{def:VirtualLevel}
The potential $V$ is \emph{critical} if the Schr\"odinger operator $H$ 
has spectrum $\sigma(H)=\sigma_{ess}(H)=[0,\infty)$ and for all nontrivial 
compactly supported potentials 
$W\ge 0$ which are infinitesimally form bounded with respect to $-\Delta+V_+$   
the family of operators $H_\lambda= H-\lambda W$  
has essential spectrum  $\sigma_{ess}(H_\lambda)= [0,\infty)$ and 
a negative energy bound state for all $\lambda>0$. 
\end{definition} 
\begin{remark}    
  The potential $V$ is called \emph{subcritical}, if 
  the Schr\"odinger operator $H$ has spectrum 
  $\sigma(H)=\sigma_{ess}(H)=[0,\infty)$ and there exist 
  a nontrivial potential $W\ge 0$, which is  infinitesimally 
  form bounded  with respect to 
  $-\Delta+V_+$,    such that    
  $H -\lambda W\ge 0$ for some  $\lambda>0$. 
\end{remark}

\begin{theorem}[Existence of a zero energy ground state for critical potentials]\label{thm:intro existence} 
  Assume that the potential $V$ satisfies Assumption \ref{assumption} and that  it is critical.
  If for some $m\in\N_0$, $\epsilon>0$, and $R>e_m$ 
  \begin{equation}\label{eq:intro existence}
    V(x) \geq  
			\frac{d(4-d)}{4|x|^2} 
				+ \frac{1}{|x|^2}\sum_{j=1}^{m}\prod_{k=1}^j\ln_k^{-1}(|x|) 
				+ \frac{\epsilon}{|x|^2}\prod_{k=1}^m\ln_k^{-1}(|x|)
  \end{equation}
  for all $|x|\ge R$, then zero is an eigenvalue of $H$. 
\end{theorem}
\begin{remark}
Clearly, the right hand sides of  \eqref{eq:intro absence} and \eqref{eq:intro existence} are, for each fixed $n\in\N$ complementary. Thus our criteria for existence and non-existence of zero energy ground states at the edge of the essential spectrum are sharp! 
Considering the simplest case $m=0$ we have  
\begin{equation}\label{eq:absence-m zero}
	V(x)\le  \frac{d(4-d)}{4|x|^2}
\end{equation}
for the absence and 
\begin{equation}\label{eq:existence-m zero}
	V(x)\ge  \frac{d(4-d)+\epsilon}{4|x|^2}
\end{equation}
for the existence with $\epsilon>0$ and all $|x|$ large enough. 
For $d=3$ this recovers the results proved in 
\cite{GriGar07} for the special case of three dimensions. 

Using the higher order corrections from Eqs.~\eqref{eq:intro absence} and \eqref{eq:intro existence} we obtain a sharp distinction between existence and non-existence in the case of a critical potential. 
For example, the cases $m=1,2$ show that if 
\begin{align}
	V(x)&\le  \frac{d(4-d)}{4|x|^2}	+\frac{1}{|x|^2\ln |x|} \quad \text{ or } \label{eq:intro absence m=1}\\
	V(x)&\le  \frac{d(4-d)}{4|x|^2}	+\frac{1}{|x|^2\ln |x|} 
	+\frac{1}{|x|^2\ln |x|\ln_2 |x|} \label{eq:intro absence m=2}
\end{align}
for all large enough $|x|$, then zero will not be a ground state eigenvalue. Conversely, for critical potentials the bound   
\begin{align}
	V(x)&\ge  \frac{d(4-d)}{4|x|^2}+\frac{1+\epsilon}{|x|^2\ln|x|} \quad \text{ or } 
	\label{eq:intro existence n=1} \\
	V(x)&\ge  \frac{d(4-d)}{4|x|^2}+\frac{1}{|x|^2\ln |x|} +\frac{1+\epsilon}{|x|^2\ln |x|\ln_2|x|} 
	 \label{eq:intro existence n=2} 
\end{align}
for all large enough $|x|$ and some $\epsilon>0$ implies that zero 
is a ground state eigenvalue. 
Using $d=3$ in \eqref{eq:intro absence m=1} 
recovers the non-existence result of \cite{BenYar90}. 
The $d=3$ case in \eqref{eq:intro existence n=1} provides a complementary existence result which was missing in \cite{BenYar90}. 

More importantly, our results provide, to arbitrary order, 
a whole family of complementary sharp criteria which 
are not restricted to three dimensions and our proofs 
are considerably simpler than the approaches based 
on delicate estimates for Green's functions.

  One often says that a Schr\"odinger operator with a critical potential has a \emph{virtual level}  (at zero energy), see, e.g. \cite{BarBit19, BarBit20, BarBitVug19-a}.  
	Theorem \ref{thm:intro absence} shows that such a virtual level is \emph{not a bound state} of 
	$H$ if $V$ obeys the bound \eqref{eq:intro absence}, that is, it is a so--called zero energy resonance. 
	Conversely, Theorem \ref{thm:intro existence} shows that a virtual level is an eigenvalue \emph{at the edge of the essential spectrum} when the potential $V$ satisfies the 
   \emph{complementary bound} \eqref{eq:intro existence}.  
\end{remark}
\begin{remark}
  In Appendix \ref{sec:appendix} we construct a family 
  of potentials $V_{\alpha,d}$ on $\R^d$ for $\alpha\in\R$ and $d\in\N$ such that the  
  Schr\"odinger operator $H_{\alpha,d}= -\Delta +V_{\alpha,d}$ has spectrum 
  $\sigma(H_{\alpha,d})=[0,\infty)$. Moreover,  
  $V_{\alpha,d}$ is subcritical for 
  $\alpha<0$ and critical for $\alpha\ge 0$. 
  The Schr\"odinger operator $H_{\alpha,d}$ has  
  a zero energy resonance for 
  $0\le \alpha\le 1$ and a zero energy bound state 
  for $\alpha>1$ in any dimension.     
\end{remark}

\begin{remark}\label{rem:essential spectrum} 
  The operator $H_\lambda$ is well-defined with quadratic form methods for all $\lambda$, 
  see Remark \ref{rem:construction H-lambda}.    
 In order to guarantee that $\sigma_{ess}(H_\lambda)=[0,\infty)$ 
 in Definition~\ref{def:VirtualLevel}, some   
 decay of the potential $V$ is required. 
A well-known sufficient criteria for this is that $V$ is relatively 
form compact with respect to the kinetic energy  $P^2=-\Delta$, 
see  \cite{Tes14}. 
This also implies that $V$ is infinitesimally form bounded, i.e.,   
relatively form small with relative bound zero, w.r.t.\ $P^2=-\Delta$, 
which excludes Hardy type potentials.  

A \emph{much less restrictive criterium} for 
$\sigma_{ess}(H)=[0,\infty)$  only assumes that $V$ \emph{vanishes asymptotically with respect 
to the kinetic energy}. More precisely, if  
\begin{equation}\label{vanishing}
	|\la \varphi, V\varphi \ra| \le a_n\|\nabla\varphi\|^2 +b_n\|\varphi\|^2
\end{equation}
for all $\varphi\in H^1(\R^d)$ with support $\supp(\varphi)\subset \{|x|\ge R_n\}$ for some sequences  $0\le a_n,b_n\to 0$, and $R_n\to\infty$ as $n\to\infty$, then $\sigma_{ess}(H)=[0,\infty)$,  see \cite{AvrHunHyn,JorWei73}. 

This criterion is clearly in line with the physical heuristic that \emph{only the asymptotic behavior of the potential 
near infinity} determines the essential spectrum and it allows for 
strongly singular potentials which are not infinitesimally form 
bounded. 
It also shows that $\sigma_{ess}(H_\lambda) = \sigma_{ess}(H)=[0,\infty)$ for all $\lambda>0$ 
when $W$ has compact support and is infinitesimally form bounded w.r.t.\ $-\Delta$ and $V$ is form small w.r.t. $-\Delta$ and satisfies \eqref{vanishing}.    
\end{remark} 

\begin{remark}
  The bounds on the potential in Theorems \ref{thm:intro absence} and \ref{thm:intro existence} are similar in spirit to logarithmic corrections to the Hardy inequality. For 
  $\psi\in \calC^\infty_0\big(\R^d\setminus\{|x|<e_m\}\big)$ and 
  one has 
  \begin{equation}
  	\la \nabla\psi,\nabla\psi \ra 
  		\ge 
  			\big\la  \psi,    
  				\Big(\frac{(d-2)^2}{4|x|^2}	+\frac{1}{4|x|^2}\sum_{j=1}^m\prod_{k=1}^j\ln_k^{-1}(|x|)\Big) 
			\psi\big\ra\, , 
  \end{equation}
  see \cite{PerCouMal86}, which also discusses conditions on the 
  potential such that $-\Delta+V$ has infinitely many, respectively 
  finitely many negative eigenvalues. Bounds on the number of 
  negative eigenvalues are given in \cite{Lun09, Lun-thesis10}. 	
  Certain logarithmic refinements of Hardy's inequality  have been used 
  to study the existence of resonances of Schr\"odinger operators and 
  the Efimov effect in low dimensions, see \cite{BarBitVug19-b}.  
\end{remark}

\begin{remark} 
Theorems \ref{thm:intro absence} and \ref{thm:intro existence} show a \emph{spectral phase transition} concerning the existence 
of zero energy ground states for Schr\"odinger operators with critical dimension $d=4$: 

The sign of the leading order term in \eqref{eq:intro absence} and \eqref{eq:intro existence} strongly depends  
on the dimension $d$, being positive if $d\le 3$, zero in dimension $d=4$, and negative if $d\ge 5$. 

Moreover, in dimension $d=4$, the leading order term vanishes and the next leading order term with $m=1$ becomes dominant. 
Thus the four dimensional case is \emph{critical}. 
Nevertheless, since the new leading order term for $d=4$ is 
also positive,  the four dimensional case is  similar to the 
case of lower dimensions. In particular, non--positive potentials 
$V$ cannot support zero energy ground states in dimensions 
$d\le 4$ while in dimension $d\ge 5$ non--positive critical potentials 
have zero energy ground states. 

Hence non--positive critical potentials will always create resonances 
in dimension $d\le 4$, while in dimension $d\ge 5$ they  
have zero energy ground states unless their negative part is 
so long range such that the bound \eqref{eq:intro existence} 
does not hold anymore. 
Nevertheless, in dimension $d\le 4$ a `long-range' positive tail of 
the potential can create zero energy ground states. 
See also the discussion in Section~2 of \cite{HunJexLan21-Helium}.

In particular, assume that the potential $V$ is infinitesimally 
form bounded w.r.t.\ $-\Delta$ and has compact support.  
Then $\sigma_{ess}(-\Delta +\beta V)=[0,\infty)$ for 
all $\beta\ge 0$, see \cite{AvrHunHyn,Tes14}, and a simple application of the min--max principle shows 
that as soon as negative eigenvalues of $-\Delta+\beta V$ exist, 
they are decreasing in $\beta>0$, see \cite[Proposition after Theorem XIII.2, page 79]{ReeSim4}. 
Let $\beta_0> 0$ be the value of the coupling constant 
when the ground energy of $-\Delta+\beta V$ hits zero.  
Theorem \ref{thm:intro absence} shows that  $-\Delta+\beta_0 V$ has 
a zero energy resonance when $d\le  4$ and Theorem \ref{thm:intro existence} shows that it has a zero energy ground state in 
dimension $d\ge 5$.  
The asymptotic of the eigenvalues of the perturbed operators 
$-\Delta+\beta V$ in $\beta-\beta_0$ 
was studied in \cite{KlSi-1} for all dimensions. 
\end{remark}
The structure of our paper is as follows: 
In Section~\ref{sec:DefAndResult} we present all the necessary 
technical tools to precisely formulate our main results. 
Theorem \ref{thm:intro absence} is proven in Section 
\ref{sec:proof non-existence}. 	The proof is by contradiction, 
assuming that a zero energy ground state exists and then deriving 
a lower bound which shows that it cannot be square integrable. 
To construct such a lower bound one only needs to know that 
a ground state, if it exists, can be chosen to be positive and that 
it is locally bounded away from zero. 
It is well--known that ground states of a Schr\"odinger operator 
$H$ in $L^2(\R^d)$ are unique, up to global phase, and can be 
chosen to be strictly positive as soon as they exist, 
see \cite{Far72, Goe77} or \cite[Section XIII.12]{ReeSim4}. 
Thus, if one knows that the 
ground state is bounded away from zero,  
 one can relax the condition on $V$ to $V\in L^1_{loc}(\R^d)$ 
and $V_-$ satisfies \eqref{eq:form small}, The assumption that $V$ is in the local Kato--class is only needed 
to guarantee that eigenfunctions of $H$ are continuous, 
see \cite{AizSim82, Sim82} and also \cite{Simader90}. 
This continuity then  guarantees that the positive ground state is 
bounded away from zero on compact sets.  	    

The proof of Theorem \ref{thm:intro existence} is given in 
Section \ref{sec:existence}. The main tool is an upper bound for 
the spacial decay of ground states of the approximating 
Schr\"odinger operators $H_\lambda$, 
see Definition \ref{def:VirtualLevel},  which is \emph{uniform} 
in $\lambda>0$. 

Since it will be necessary to have a positive ground state for 
the non-existence proof,  we cannot prove the absence of ground state under symmetry constraints which destroy the 
positivity of ground states, such as fermionic particle statistics. 
However, the existence proof still works under symmetry restrictions, 
see Remark \ref{rem:existence under symmetry constraints}.

In Appendix \ref{sec:appendix} we construct an 
explicit example of a family of potentials which 
exhibits all possible different scenarios. 
\smallskip

Lastly, recall that the Kato--class $K_d$ is given by all real--valued 
potentials $V$ such that in dimension $d\ge 2$
\begin{equation}\label{eq:Kato class} 
  \lim_{\alpha\downarrow0}\sup_{|x|\in \R^d}\int_{|x-y|\leq\alpha} g_d(x-y)|V(y)|dy=0\,,\\
\end{equation}
where 
\begin{equation}\label{eq:g-d}
  g_d(x)
  	\coloneqq 
  		\left\{ 
  			\begin{array}{cll}
  				|x|^{2-d} & \text{if} & d\ge 3 \\
  				|\ln|x||  & \text{if} & d= 2 
  			\end{array}
  		\right. \, .
\end{equation}
The Kato class in one dimension is given by  $K_1\coloneqq L^1_{loc,unif}(\R)$, the space of uniformly locally integrable functions on $\R$. 
We say that the potential $V$ is in the local Kato--\-class $K_{d,loc}$  if 
$V \ind_K \in K_d$ for all compact sets  
$K\subset\R^d$. 
It is clear that $K_d\subset L^1_{loc,unif}(\R^d)$ and   
$K_{d,loc}\subset L^1_{loc}(\R^d)$. Moreover, it is 
well--known that any potential $V\in K_d$ is infinitesimally form small with respect 
to $-\Delta$, see \cite{CycFroKirSim87}. 

Thus if $V=V_+ - V_-$ with $V_\pm\ge 0$,  $V_+\in K_{d,loc}$, 
and $V_-\in K_d$ then all of the claims of Assumption \ref{assumption} hold. 
This class of potentials  is large enough to include most, if not all, 
physically relevant potentials, except maybe for some highly oscillatory potentials.

\section{Definitions and preparations}\label{sec:DefAndResult}
Assume that  $V_\pm\in L^1_{loc}(\R^d)$ and \eqref{eq:form small} holds for $V_-$.  
The KLMN theorem \cite{ReeSim4,Tes14} then shows 
that there exists a unique  self-adjoint operator $H$ 
corresponding to a quadratic form   
  \begin{equation}\label{eq:quadratic form H}
  	\la \psi, H\psi \ra \coloneqq
  		\la \nabla\psi,\nabla\psi\ra 
  			+ \la \sqrt{V_+}\psi, \sqrt{V_+}\psi\ra 
			- \la \sqrt{V_-}\psi, \sqrt{V_-}\psi\ra \, 
  \end{equation}
with the usual slight abuse of notation. Here 
$\psi\in \calQ(H)\coloneqq H^1(\R^d)\cap\calQ(V_+)$, the form domain of $H$, where $H^1(\R^d)$ is the usual $L^2$ based Sobolev space 
  of functions $\psi\in L^2(\R^d)$ whose weak 
  (distributional) gradient $\nabla\psi\in L^2(\R^d)$, and 
  \begin{equation}\label{eq:calQ(V)}
  	\calQ(V_+)\coloneqq\calD(\sqrt{V_+}) 
  		= \big\{ \psi\in L^2(\R^d):\,  \sqrt{V_+}\psi\in L^2(\R^d) \big\}
  \end{equation}
  is the quadratic form domain of the multiplication operator $V_+$. 
  
  Since $\sqrt{V_+}\in L^2_{loc}$ we clearly have $\calC^\infty_0(\R^d)\subset \calQ(H)$. 
  Note that 
   $\calC^\infty_0(\R^d)$ is a form core, i.e., dense in  $H^1(\R^d)\cap \calD(\sqrt{V_+})$ with 
  respect to the norm 
  \begin{equation}\label{eq:quadratic form norm}
 \|\psi\|_1\coloneqq (\|\psi\|_{H^1}^2+ \|\sqrt{V_+}\psi\|^2)^{1/2} \, , 
  \end{equation} 
  see \cite{CycFroKirSim87, LeiSim81}. In addition, Friedrich's extension theorem, see for example \cite[Theorem 2.13]{Tes14},  
  implies that the operator $H$ and its domain $\calD(H)$ are 
  given by 
  \begin{equation}\label{eq:H}
  	\begin{split}
  	D(H) &= \big\{ \psi\in H^1(\R^d)\cap\calQ(V_+):\,  (-\Delta+V)_{\distr}\,\psi\in L^2(\R^d) \big\} \\
	H\psi &= (-\Delta+V)_{\distr}\,\psi
  	\end{split}
  \end{equation}
  where $(-\Delta+V)_{\distr}\psi$ is in the sense of distributions when acting on $\psi\in L^2(\R^d)$.  
  \smallskip
  
\begin{remark}\label{rem:construction H-lambda}
	If  $V, W\in L^1_{loc}$ and $V_-$ is form small and $|W|$ is form bounded with respect to 
	$-\Delta+V_+$, i.e., \eqref{eq:form small} holds for $V_-$ for some $0\le a_1<1$, $b\ge 0$ and it also holds with  
	$V_-$ replaced by $|W|$ for some $a_2,b_2\ge 0$,  then 
	\begin{equation}\label{eq:form small 2}
		\|\sqrt{V_-}\psi\|^2+\lambda\|\sqrt{|W|}\psi\|^2
			\le 
				(a_1+\lambda a_2)\big(\|\nabla\psi\|^2 + \|\sqrt{V_+}\psi\|^2 \big) + (b_1+\lambda b_2)\|\psi\|^2
	\end{equation}
  for all $\psi\in H^1\cap \calQ(V_+)$. 
  So for any $0<\lambda_0<(1-a_1)/a_2$, we can construct the family 
  of Schr\"odinger operators $H_\lambda$ as the 
  unique self-adjoint operator given by the quadratic forms 
  \begin{equation}\label{eq:quadratic form H delta}
	\la \psi,H_\lambda\psi\ra 	
		\coloneqq 
			\la \nabla\psi,\nabla\psi\ra  + \la \psi, V_+\psi\ra -\la \psi,V_-\psi\ra 
				-\lambda \la \psi, W\psi\ra \, .
  \end{equation} 
  with quadratic form domain $\calQ(H_\lambda) = H^1(\R^d)\cap \calQ(V_+)= \calQ(H)$
  for all $0\le \lambda \le \lambda_0$. For $\lambda=0$ one recovers $H$. If $W$ is infinitesimally form 
  bounded w.r.t.\  $-\Delta+V_+$, then $\lambda_0=\infty$. 
\end{remark}  
 \smallskip 
   
  One can relax the conditions on $V$ to hold only on a connected, open  
  set $U\subset\R^d$, which contains infinity. 
  In this case one assumes $V_+\in L^1_{loc}(U)$, and \eqref{eq:form small} holds for all 
  $\psi\in H^1_0(U)\cap\calQ^U(V_+)$, where $H^1_0(U)$ 
  is the usual Sobolev space with Dirichlet boundary 
  conditions on the boundary $\partial U$ and 
  $\calQ^U(V_+)= \{\psi\in L^2(U):\, \sqrt{V_+}\psi\in  L^2(U)\}$. 
  In this case $H$ is the Schr\"odinger operator (with Dirichlet boundary 
  conditions) defined by the quadratic form \eqref{eq:quadratic form H} 
  which is restricted to $\psi\in \calQ^U(H)= H^1_0(U)\cap\calQ^U(V_+)$. 
  Again it is well known that $\calC^\infty_0(U)$ is 
  dense in $\calQ^U(H)$ w.r.t.~ the norm given in \eqref{eq:quadratic form norm}. 
  The same holds for $H_\delta$ and any $\delta>0$ small enough. 

\smallskip
Now assume  that the real-valued 
potential $V\in L^1_{\loc}(\R^d)$, that its 
negative part $V_-$ is form small w.r.t, $-\Delta+V_+$, i.e., \eqref{eq:form small} holds, and let $H$ be the associated Schr\"odinger operator defined by quadratic form methods as above.   
For an open set $U\subset\R^d$ we consider weak (local) 
eigenfunctions of $H$ at energy $E$, i.e., (weak local) solutions 
of the Schr\"odinger equation
\begin{equation}\label{eq:eigenfunction}
	H \psi = E\psi\quad\text{in } U \, .
\end{equation}
We are mainly interested in the case that $E=0$.  

With a slight abuse of notation, we denote by $\la\varphi,H\psi\ra$ 
the sesquilinear form given by 
\begin{equation}\label{quadratic form}
	\la\varphi,H\psi\ra 
		\coloneqq \la \nabla\varphi, \nabla\psi\ra +\la \varphi, V\psi \ra 
		= \int(\overline{\nabla\varphi}\cdot\nabla\psi +\overline{\varphi} V \psi)\, dx
\end{equation}
whenever the right hand sides makes sense. This is the case if $\varphi,\psi\in \calQ(H)= H^1(\R^d)\cap\calQ(V_+)$ but also if, 
$\varphi\in \calQ^U_c(H)$ and $\psi\in \calQ^U_{loc}(H)$, where for some open set $U\subset\R^d$ 
the local quadratic form domain   
\begin{equation}
	\calQ^U_{loc}(H)
		\coloneqq \big\{ \psi\in L^2_{\loc}(U):\, \chi\psi\in \calQ(H) \text{ for all } \chi\in \calC^\infty_0(U)\big\}\, 
\end{equation}
is the vector space of functions which are locally (in $U$) 
in the quadratic form domain of $H$. Moreover,
\begin{equation}
	\calQ^U_c(H)
		\coloneqq \big\{ \psi\in \calQ(H):\, 
			\supp(\psi)\subset U \text{ is compact  }\big\}\, 
\end{equation}
is the set of functions in $\calQ(H)$ with compact support inside $U$.  
If $\varphi\in \calQ^U_c(H)$ and $\psi\in \calQ^U_{loc}(H)$, then the integral on the 
right--hand--side of \eqref{quadratic form} can be restricted to the set $U$.  
Clearly, $\calQ^U_{loc}(H)= \{ \psi\in H^1_{loc}(U):\, \sqrt{V_+}\psi\in L^2_{loc}(U)\}= H^1_{loc}(U)\cap \calQ^U_{loc}(V_+)$. 

Similarly, one can define the local domain of $H$, relative to some open set $U\subset\R^d$,  by 
\begin{equation}\label{eq:local domain}
  \calD^U_{loc}(H)
  	\coloneqq 
  		\big\{ 
  			\psi\in L^2_{\loc}(U):\, \chi\psi\in \calD(H) 
  			\text{ for all } \chi\in \calC^\infty_0(U) 
  		\big\}\, . 
\end{equation}  
\begin{remark}\label{rem:local spaces}
  Note that the definitions of $\calQ^U_{loc}(H)$ and 
  $\calD^U_{loc}(H)$ are consistent in the sense that 
  for any $\chi\in \calC^\infty_0(U)$ one has  
  $\chi\psi\in \calQ(H)$ {\rm(}even $\chi\varphi\in \calQ^U_c(H)${\rm)} for any $\psi\in \calQ(H)$ and 
  $\chi\psi\in \calD(H)$ for any $\psi\in \calD(H)$.  
  This is clear when 
  $\psi\in \calQ^U_{loc}(H)= H^1_{loc}(U)\cap\calQ^U_{loc}(V_+)$ 
  since for $\chi\in \calC^\infty_0(U)$ we have 
  $\chi\psi\in H^1(\R^d)$ for any 
  $\psi\in H^1_{loc}(U)$ and $\chi\psi \in \calQ(V_+)$
  for any $\psi\in \calQ^U_{loc}(V_+)$. In addition, if 
  $\psi\in \calD(H)$ then 
  \begin{equation*}
  	(-\Delta+V)_{\distr} \, \chi\psi 
  		= \chi (-\Delta+V)_{\distr} \, \psi 
  		  - 2\nabla \chi\nabla\psi- (\Delta\chi)\psi  \in L^2(\R^d)\, , 
  \end{equation*}
  so $\chi\psi\in \calD(H)$. Moreover,  
  with 	$\calC^\infty(U)$ the infinitely 
  differentiable functions on $U$, it is easy 
  to see that 
  \begin{equation}
  	\calC^\infty(U)\subset \calQ^U_{loc}(H)\, . 
  \end{equation}
  since $\calC^\infty_0(U)\subset \calQ^U(H)$. 
  However, the inclusion $\calC^\infty(U)\subset \calD^U_{loc}(H)$ is wrong in general, since the construction of the Schr\"odinger operator $H$ with the help of quadratic forms allows for rather singular potentials~V. 
\end{remark}

\smallskip
Thus we define weak solutions, supersolutions and subsolutions of \eqref{eq:eigenfunction} in the following quadratic form sense.  
\begin{definition}\label{def:eigenfunctions etc}
\begin{SL} 
\item $u$ is a (weak) eigenfunction of the Schr\"odinger operator 
		$H$ with energy $E$ if  $u\in \calQ(H)$ and  
  \begin{equation}\label{eq:weak eigenfunction}
    \la \varphi, (H-E) u\ra = 0 	
  \end{equation}
  for every $\varphi\in  C_0^\infty(\R^d)$.
\item $u$ is a (weak) local eigenfunction of the Schr\"odinger operator 
		$H$ with energy $E$ in $U\subset\R^d$  
 	 if $u\in \calQ^U_{loc}(H)$ and  
  \begin{equation}\label{eq:weak local eigenfunction}
    \la \varphi, (H-E) u\ra = 0 	
  \end{equation}
  for every $\varphi\in  C_0^\infty(U)$.
\item  $u$ is a supersolution of the Schr\"odinger operator 
		$H$ with energy $E$ in $U\subset\R^d$ if 
  $u\in \calQ^U_{loc}(H)$ and 
  \begin{equation}\label{eq:supersolution}
  \la \varphi, (H-E) u\ra \ge 0 	
  \end{equation}
  for every nonnegative $\varphi\in C_0^\infty(U)$. 
\item $u$ is a subsolution of the Schr\"odinger operator 
		$H$ with energy $E$ in $U\subset\R^d$  if 
  $u\in \calQ^U_{loc}(H)$ and 
  \begin{equation}\label{eq:subsolution}
  \la \varphi, (H-E)  u\ra \le  0 	
  \end{equation}
  for every nonnegative $\varphi\in  C_0^\infty(U)$. 
\end{SL}
\end{definition}
\begin{remark}
  Using the density of $\calC^\infty_0$ in $\calQ(H)$ it is easy to see that once \eqref{eq:weak eigenfunction} 
  holds for all $\varphi\in \calC^\infty_0$, it holds for all $\varphi\in \calQ(H)$. Similarly, 
  \eqref{eq:weak local eigenfunction} holds for all 
  $\varphi\in \calQ^U_c(H)$, and   
  \eqref{eq:supersolution}, respectively \eqref{eq:subsolution}, hold for all 
  $0\le \varphi\in \calQ^U_c(H)$.  
\end{remark}
One should note that one does not have to distinguish between weak eigenfunctions and eigenfunctions and similarly for local eigenfunctions. 

\begin{lemma}\label{lem:regularity}
  Every weak eigenfunction $u\in \calQ(H)$ of $H$ is in $\calD(H)$ given by \eqref{eq:H}.  
  Similarly, if $u\in \calQ^U_{loc}(H)$ is a weak local eigenfunction of $H$ 
  in an open domain $U\subset\R^d$, then $u$ is 
  locally in the domain of $H$, i.e., $u\in \calD^U_{loc}(H)$ given by \eqref{eq:local domain}.
\end{lemma}
\begin{proof}
  This is probably a standard argument for weak eigenfunctions, 
  but not standard for weak local eigenfunctions. 
  Let $f\in L^2(\R^d)$ and $\psi$ be a weak solution of the equation $H\psi = f$, i.e., 
  \begin{equation}\label{eq:resolvent 1}
  	\la \varphi, H\psi \ra = \la \varphi, f \ra 
  \end{equation}
  for all $\varphi\in\calQ(H)$. 
  Then for any 
  $\lambda\in \R$ we have 
  \begin{equation}\label{eq:resolvent 2}
  	\la \varphi, (H+\lambda)\psi \ra = \la \varphi, \lambda\psi+ f \ra 
  \end{equation}
  for all $\varphi\in\calQ(H)$. Since $H$ is bounded from below, all large enough $\lambda$ will be in the resolvent set of $H$. So for all large enough 
  $\lambda$ we can choose $\varphi= (H+\lambda)^{-1}\xi$, with $\xi\in L^2(\R^d)$ in \eqref{eq:resolvent 2} to get 
    \begin{equation}\label{eq:resolvent 3}
  	\la \xi, \psi \ra 
  		= \la (H+\lambda)^{-1}\xi, \lambda\psi+ f \ra 
  		= \la \xi, (H+\lambda)^{-1}(\lambda\psi+ f) \ra\, . 
  \end{equation} 
  This holds for all $\xi\in L^2(\R^d)$, so 
  \begin{equation}\label{eq:resolvent 4}
  	\psi 
  		= (H+\lambda)^{-1}(\lambda\psi+ f) \in \calD(H)\, ,
  \end{equation} 
  since $\psi, f\in L^2(\R^d)$ and the resolvent 
  $(H+\lambda)^{-1}$ maps $L^2(\R^d)$ onto $\calD(H)$. 
  
  Note that if $\psi$ is a weak eigenfunction of $H$, 
  at energy $E$, then we can use  $f= E\psi$. 
  Thus weak eigenfunctions are eigenfunctions in 
  the domain of $H$. 
  
  Now assume that $f\in L^2_{loc}(U)$ and 
  $\psi\in \calQ^U_{loc}(H)$ is a weak local solution of
  \begin{equation}\label{eq:resolvent 1 local}
  	\la \varphi, H\psi \ra = \la \varphi, f \ra 
  \end{equation}
 for all  $\varphi\in \calQ^U_c(H)$. Take any 
  $\chi\in\calC^\infty_0(U)$. Replacing $\varphi$ by 
  $\ol{\chi}\varphi$ in \eqref{eq:resolvent 1 local} one 
  sees that 
  \begin{equation}\label{eq:resolvent 2 local}
  	\la \ol{\chi}\varphi, H\psi \ra 
  		= \la \varphi, \chi f \ra 
  \end{equation}
  for all $\varphi\in\calQ(H)$. Using that $\chi, \nabla\chi$, and $\Delta\chi$ have compact supports, a straightforward calculation shows 
  \begin{equation*}
  	\la \nabla(\ol{\chi}\varphi), \nabla\psi\ra
  		= 
  			\la \nabla\varphi, \nabla(\chi\psi) \ra
  			+ \la \varphi, (\Delta\chi +2\nabla\chi\nabla) \psi \ra\, .
  \end{equation*} 
  Using this and the definition \eqref{eq:quadratic form H} of the quadratic form  
   in \eqref{eq:resolvent 2 local} yields 
  \begin{equation}\label{eq:resolvent 3 local}
  	\la \varphi, H\chi\psi \ra 
  		= \la \varphi, \chi f -   (\Delta\chi +2\nabla\chi\nabla) \psi\ra 
  \end{equation}
  for all $\varphi\in \calQ(H)$. Adding again $\la \varphi,\lambda\chi\psi\ra$ on both sides and choosing $\varphi=(H+\lambda)^{-1}\xi$ with 
  $\lambda$ large that enough, one sees that   
  \begin{equation}\label{eq:resolvent 4 local}
  	\la \xi, \chi\psi \ra 
  		= \la \xi, (H+\lambda)^{-1}\big(\chi(\lambda\psi+ f) -   (\Delta\chi)\psi +2\nabla\chi\nabla\psi\big)\ra 
  \end{equation}
  for all $\xi\in L^2(\R^d)$. Hence 
  \begin{equation}\label{eq:resolvent 5 local}
  	\chi\psi 
  		= (H+\lambda)^{-1}\big(\chi(\lambda\psi+ f) -   (\Delta\chi)\psi +2\nabla\chi\nabla\psi\big) \in \calD(H)\, ,
  \end{equation}
  for any $\chi\in \calC^\infty_0(U)$. 
  Thus $\psi\in \calD^U_{loc}(H)$. 
  Again, using $f=E\psi$ shows that any weak local eigenfunctions of $H$ at energy $E$ 
  is locally in the domain of $H$. 
\end{proof}

Finally let us note that the definition of critical potential and virtual levels are rather natural. 
It is easy to see that any potential which creates a zero energy ground 
state is critical. 
\begin{lemma}\label{lem:criticalV}
Assume that $V\in L^1_{loc}(\R^d)$ and that $V_-$ is form small 
 and $0\le W$ is infinitesimally form small w.r.t.\ $-\Delta+V_+$. Furthermore let $H$ and $H_\lambda$, $0 < \lambda \le  \lambda_0$, 
 be the associated Schr\"odinger operators, see Remark \ref{rem:construction H-lambda}.  
Assume also that  $\sigma(H)= \sigma_{ess}(H_\lambda)= [0,\infty)$ and that   
$H$ has a zero energy ground state.   
Then the potential  $V$ is critical. 
\end{lemma}
\begin{proof} 
This is probably well--known. We provide the short proof for the convenience of the reader. 
Let $\psi$ be a zero energy normalized ground state of $H$. Then for any small enough $\lambda>0$.  
\begin{equation*}
	\la \psi, H_\lambda\psi\ra = \la \psi, H\psi\ra -\lambda\la\psi,W\psi\ra = -\lambda\la \psi, W\psi\ra<0 
\end{equation*} 
since we can choose the ground state $\psi>0$. 
Thus as soon as 
$\sigma_{ess}(H_\delta)= [0,\infty)$, the min--max principle shows that 
$H_\delta$ has eigenvalues below zero. 
\end{proof}

The converse to Lemma \ref{lem:criticalV} does not hold. 
See the example from Appendix \ref{sec:appendix}. 

\smallskip

Our proofs of Theorems \ref{thm:intro absence} and 
\ref{thm:intro existence} rely on the so--called subharmonic comparison lemma which has already 
seen wide use in the study of the asymptotic decay of eigenfunctions 
of Schr\"odinger operators, see, e.g., \cite{DeiHunSimVoc78,HofOstHofOstAhl78}. 
We use the version of \cite[Theorem 2.7]{Agm85} since 
it allows for a quadratic form version which needs only minimal 
regularity assumptions. 

\begin{theorem}[Agmon's version of the comparison principle]\label{thm:comp-agm}
Let $w$ be a positive supersolution of the Schr\"odinger operator $H$ at energy $E$ in a neighborhood of infinity 
$U_R\coloneqq\{x \in \R^d\,:\, |x| > R\}$. Let $v$ be a subsolution of 
$H$ at energy $E$ in $U_R$. Suppose that
\begin{equation}\label{eq:liminf cond}
\liminf_{N\rightarrow\infty}\left(\frac{1}{N^2}\int_{N\leq|x|\leq\alpha N}|v|^2\mathrm d x\right)=0
\end{equation}
for some $\alpha>1$. If for some $\delta>0$ and $0\le C<\infty$ one has 
\begin{equation}\label{eq:a-priori comparison on annulus}
v(x)\leq Cw(x) \,  \text{ on the annulus } R<|x|\le R+\delta \, ,
\end{equation}
 then 
\begin{equation}
v(x)\leq Cw(x) \, \text{ for all } x\in U_{R}\,.
\end{equation}
\end{theorem}

\begin{remark}
We note that the condition \eqref{eq:liminf cond} is trivially satisfied 
as soon as  $v\in L^2(\R^d)$, but it also allows for subsolutions $v$ which are not square integrable at infinity.  
This is crucial for the proof of our non-existence result.   
A slight extension of Agmon's comparison principle, which allows to relax the continuity assumptions and works for domains $U$ which are not necessarily neighborhoods of infinity,  is derived in \cite{HunJexLan21-Helium}.   
\end{remark}

\begin{remark}
Agmon also assumes that the supersolution $w$ and the  subsolution $v$ are continuous in $\overline{U_R}= \{|x|\ge R\}$ 
in \cite{Agm85}. 
However, the form of Theorem \ref{thm:comp-agm} is what is really proven in \cite{Agm85}, see 
also \cite{HunJexLan21-Helium}.   
The additional assumption that the supersolution $w>0$ and the subsolution $v$ are  continuous in 
$\{|x|\ge R\}$  are only made in  \cite{Agm85} to guarantee that \eqref{eq:a-priori comparison on annulus} 
holds with constant $C= c_2/c_1$ where   
$c_1\coloneqq \inf_{R\le |x|\le R+\delta} w(x)>0$ and $c_2\coloneqq \sup_{R\le |x|\le R+\delta}|v(x)|<\infty$, 
by continuity, for arbitrary $\delta>0$. 
\end{remark}

Before we give the proofs of Theorems \ref{thm:intro absence} and \ref{thm:intro existence}, 
we sketch the simple proof of the $m=0$ version of Theorem \ref{thm:intro absence} using the comparison theorem:

For $\gamma>0$ set $\psi_\gamma(x)= |x|^{-\gamma}$, so  
$\psi_\gamma\in \calC^\infty(U_R)$ and all $R>0$.  
A short calculation shows $\Delta\psi_\gamma(x)= \gamma(\gamma+2-d)|x|^{-\gamma-2}$ for $x\neq 0$. 
Hence with 
\begin{equation}\label{eq:W-gamma}
	W_\gamma(x) 
		=\frac{\Delta\psi_\gamma(x)}{\psi_\gamma (x)} 
		= \gamma(\gamma+2-d)|x|^{-2}
\end{equation}
one sees that $(-\Delta +W_\gamma(x))\psi_\gamma(x)=0$ for $x\neq 0$. Since $\psi_\gamma \in \calC^\infty(\R^d\setminus\{0\})$, integration 
by parts shows that $\psi_\gamma$ is a weak local 
eigenfunction of $-\Delta +W_\gamma$ in the sense of 
Definition \ref{def:eigenfunctions etc} in the open sets $U_R$ for any $R>0$.  

Moreover, Remark \ref{rem:local spaces} shows that 
$\psi_\gamma\in \calQ^{U_R}_{loc}(H)$ for any $R>0$ 
and any Schr\"odinger operator $H$ with potential 
$V\in L^1_{loc}(\R^d)$ for which 
$V_-$ is form small w.r.t.\ $-\Delta+V_+$. Thus if 
\begin{equation}
	V(x)\le W_\gamma(x) 
\end{equation}
for some $\gamma>0$ and all large enough $|x|> R$, then
\begin{equation} \label{eq:V below W gamma}
	\la \varphi, H\psi_\gamma\ra 
	\le \la \varphi, (-\Delta +W_\gamma)\psi_\gamma\ra 
	=0 
\end{equation}
for all $0\le \varphi\in\calC^\infty_0(U_R)$, i.e., 
$\psi_\gamma$ is a zero energy subsolution of $H$. 
One easily checks that 
\begin{SL}
  \item[a)] $\psi_\gamma \not\in L^2(U_R)$ if and only if $0<\gamma\le d/2$.
  \item[b)] $\psi_\gamma$ satisfies \eqref{eq:liminf cond} if and only if $\gamma>(d-2)/2$. 
  \item[c)] $0<\gamma\mapsto W_\gamma$ is increasing if and only if $\gamma>(d-2)/2$.
\end{SL}
The last part shows that one should choose $\gamma$ 
as large as possible in order to guarantee that \eqref{eq:V below W gamma} holds.   

Now assume that $H\ge 0 $ has zero as an eigenvalue with corresponding unique ground state $\psi$ which can be chosen to be positive \cite{Far72, Goe77}.  Since $V$ is locally in the Kato class, one also knows that $\psi$ is continuous, \cite{Sim82}.
If  
\begin{equation}
	V(x)\le \frac{d(4-d)}{4|x|^2} = W_{d/2}(x)
\end{equation}
then the above discussion shows that  $u=\psi_{d/2}$ is a zero energy subsolution of $H$ which is not square integrable at infinity but for which \eqref{eq:liminf cond} holds. 
Using  
$c^1_R= \inf_{R\le |x|\le R+1}\psi(x)>0$,  
$c^2_R= \sup_{R\le |x|\le R+1} u(x)$, and 
$C=c^2_R/c^1_R$ ensures $u(x)\le C \psi(x)$
for all $R\le |x|\le R+1$. Then Theorem \ref{thm:comp-agm} shows that 
\begin{equation}
	u(x)\le C\psi(x)
\end{equation}
for all $|x|>R$, in particular, $\psi\neq L^2(\R^d)$, hence $\psi$ is not an eigenfunction. 
Thus zero is not an eigenvalue of the Schr\"odinger operator $H$, which proves the $n=0$ version of Theorem \ref{thm:intro absence}.  

\begin{remark}
  Note that for $\gamma>d/2$ the function $\psi_\gamma$ 
  is in $L^2(U_R)$ for any $R>0$, so $\gamma=d/2$ 
  is the largest possible choice in order to get 
  the non--existence result.  The higher order condition for non-existence will have to use a suitably modified choices of subsolutions at $\gamma=d/2$.  
\end{remark}
\begin{remark}
  Choosing $\gamma=(d-2)/2$ yields the Hardy potential $W_{(d-2)/2}(x)= -\frac{(d-2)^2}{4|x|^2}$. It is well known that a Schr\"odinger operator with a Hardy potential is nonnegative. 
  It is curious that for the absence of zero energy eigenfunctions the choice $\gamma=d/2$ becomes relevant.
\end{remark}
  For the existence 
  result we want to reverse the roles of 
  the eigenfunction $\psi$ and $\psi_\gamma$. 
  If 
  \begin{equation}
  	V(x) \ge \frac{d(4-d)+\varepsilon}{4|x|^2}
  \end{equation}
  for all $|x|>R$ and some $\varepsilon >0$ 
  then $\psi_\gamma$ is a zero 
  energy subsolution of $H$ in $U_R$ where $\gamma>d/2$ 
  is the unique solution of $\gamma(\gamma+2-d) = (d(4-d)+\varepsilon)/4$. 
  Arguing as above, one sees that any positive 
  zero energy ground state $\psi$ of $H$ satisfies the 
  upper bound    
  \begin{equation}
  	\psi(x) \le C\psi_\gamma(x)
  \end{equation}
  for all $|x|>R$, hence it is square integrable 
  at infinity since $\psi_\gamma\in L^2(U_R)$ 
  as soon as $\gamma>d/2$.
  Of course, this is a circular reasoning, since 
  we need the existence of a square integrable 
  bound state, or at least the existence of a 
  local zero energy bound state which 
  satisfies \eqref{eq:liminf cond}. The 
  rigorous argument uses the fact that 
  $H\ge 0$ is assumed to have a virtual level at 
  zero, so the operators $H_\delta$ have a negative energy ground states with negative energy for small $\delta>0$ 
  These ground states will converge to a zero energy 
  ground state of $H$ in the limit $\delta\to 0$, see Section \ref{sec:existence}.

\section{Proof of the non--existence result}\label{sec:proof non-existence}
Recall the iterated logarithms $\ln_n$ defined by  
$\ln_1(r)\coloneqq\ln(r)$ for $r>0$ and, for  
$r>e_j$, inductively by $\ln_{j+1}(r)\coloneqq\ln(\ln_j(r))$ when  
$j\in\N$. Here $e_1=1$ and $e_{j+1}= e^{e_j}$. 

A convenient sequence of functions at the edge of $L^2$-integrability near infinity is given by 
\begin{equation}\label{eq:comparison functions lower bound}
\psi_{\ell,m}(x)\coloneqq|x|^{-d/2}\prod_{j=1}^m\ln_j^{-1/2}(|x|)
\quad \text{for } |x|>e_m\,.
\end{equation}
As usual, the empty product is one, so 
$\psi_{\ell,0}(x) = |x|^{-d/2} =\psi_{d/2}(x)$.  
We still have $\psi_{\ell,m}\in \calC^\infty(\{|x|>e_m\})$, in particular, 
$\psi_{\ell,m}\in \calQ^{U_{R}}_{loc}(H)$, for 
$R\ge e_m$ and any 
Schr\"odinger operator constructed via quadratic form methods as in Section \ref{sec:DefAndResult}.
In order to mimic the proof sketched at the end of Section \ref{sec:DefAndResult}, we need to know 
the potential $W_m$ for which $(-\Delta+W_m)\psi_{\ell,m}=0$ in $U_{e_m}= \{|x|>e_m\}$. 
This is a bit more complicated than the previous calculation for $\psi_\gamma$. 
\begin{lemma}\label{lem:Wm}
  For any $m\in\N_0$ we have $(-\Delta+W_m)\psi_{\ell,m}=0$ 
  in $U_R=\{x\in\R^d:\, |x|>R\}$, for all large enough $R\ge e_m$, where  
  \begin{equation}\label{eq:Wm}
  \begin{split}
  W_m(x)\coloneqq & \frac{d(4-d)}{4|x|^2} 
		+\frac{1}{|x|^2}\sum_{j=1}^m\prod_{k=1}^j\ln_k^{-1}(|x|)
		+\frac{1}{4|x|^2}\left(\sum_{j=1}^m\prod_{k=1}^j\ln_k^{-1}(|x|)\right)^2\\
	&+\frac{1}{2|x|^2}\sum_{j=1}^m\sum_{l=1}^j\prod_{s=1}^l\prod_{t=1}^j\ln_s^{-1}(|x|)\ln_t^{-1}(|x|)
  \end{split}
  \end{equation}
  is well--defined for $|x|>e_m$.
\end{lemma}
\begin{proof}
  Clearly, if  $W_m=\frac{\Delta\psi_{\ell,m}}{\psi_{\ell,m}}$ then $(-\Delta+W_m)\psi_{\ell,m}=0$ in $U_R$, for all large enough $R>0$. 
  For any radial function depending only on the radius $r=|x|$ we have 
  \begin{equation}
  	\Delta\psi(x) 
  		= \partial_r^2\psi(x) + \frac{d-1}{|x|}\partial_r\psi(x)\, .
  \end{equation}
  By a straightforward but slightly tedious calculation
  one sees that 
 \begin{equation*}
  \begin{split}
    \partial_r\psi_{\ell,m}(x)
	  	&=-\psi_{\ell,m}(x)
	  		\left(
	  			\frac{d}{2|x|}+\frac{1}{2|x|}\sum_{j=1}^m\prod_{k=1}^j\ln_k^{-1}(|x|)
	  		\right)\,,\\
    \partial_r^2\psi_{\ell,m}(x)
    	&=\psi_{\ell,m}(x)
    		\left(
    			\frac{d}{2|x|}+\frac{1}{2|x|}\sum_{j=1}^m\prod_{k=1}^j\ln_k^{-1}(|x|)
    		\right)^2
    		+\psi_{\ell,m}(x)\frac{d}{2|x|^2}
    		\\
    	&\phantom{=}
    		+\psi_{\ell,m}(x)
    			\left(
    				\frac{1}{2|x|^2}\sum_{j=1}^m\prod_{k=1}^j\ln_k^{-1}(|x|)
    				+\frac{1}{2|x|^2}\sum_{j=1}^m\sum_{l=1}^j\prod_{s=1}^l\prod_{t=1}^j\ln_s^{-1}(|x|)\ln_t^{-1}(|x|)
    				\right)
\end{split}
\end{equation*}
where we used
\begin{equation*}
\begin{split}
  \partial_r \ln_1(r)
  	=\frac 1 r\,\,\, \text{and }\,\,
  \partial_r \ln_j(r)
  	=\frac{1}{\ln_{j-1}(r)}\frac{1}{\ln_{j-2}(r)}\ldots\frac{1}{\ln_{1}(r)}\frac 1 r\,.
\end{split}
\end{equation*}
Thus 
\begin{equation*}
\begin{split}
	W_m(x)= \frac{\Delta\psi_{\ell,m}(x)}{\psi_{\ell,m}(x)}
	=\, \, & \frac{d(4-d)}{4|x|^2} 
		+\frac{1}{|x|^2}\sum_{j=1}^m\prod_{k=1}^j\ln_k^{-1}(|x|) 
		+\frac{1}{4|x|^2}\left(\sum_{j=1}^m\prod_{k=1}^j\ln_k^{-1}(|x|)\right)^2\\
	& +\frac{1}{2|x|^2}\sum_{j=1}^m\sum_{l=1}^j\prod_{s=1}^l\prod_{t=1}^j\ln_s^{-1}(|x|)\ln_t^{-1}(|x|)
\end{split}
\end{equation*}
and we have $(-\Delta+W_m)\psi_{\ell,m}=0 $ in 
$U_R$ as long as $R>0$ is large enough, so that the iterated logarithms are well-defined. 
\end{proof}
\begin{remark}
Alternatively, one can compute  $ W_m=\frac{\Delta\psi_{\ell,m}}{\psi_{\ell,m}}$ inductively. 
For radial functions $f,g$, i.e., with the usual abuse of notation $f(x)=f(r)$ and $g(x)=g(r)$ for $r=|x|$ we have 
\begin{equation}\label{eq:product rule}
\Delta(g f)= f \Delta g+2 \partial_r f \partial_r g+(\Delta f) g\,. 
\end{equation}
Using \eqref{eq:product rule} and $\psi_{\ell,m+1}(x)=\psi_{\ell,m}(x)\ln_{m+1}^{-\frac 1 2}(|x|)$  we obtain
\begin{equation*}
  W_{m+1}(x)=\frac{\Delta\psi_{\ell,m+1}(x)}{\psi_{\ell,m+1}(x)}
    = W_m + \frac{\Delta \ln_{m+1}^{-\frac 1 2}(|x|)}{\ln_{m+1}^{-\frac 1 2}(|x|)}
    	+ 2\frac{\partial_{r}\psi_{\ell,m}(x)}{\psi_{\ell,m}(x)}
    		\frac{\partial_{r}\ln_{m+1}^{-\frac 1 2}(|x|)}{\ln_{m+1}^{-\frac 1 2}(|x|)}\,.
\end{equation*}
A straightforward calculation yields
\begin{equation*}
\begin{split}
  \frac{\Delta \ln_{m+1}^{-\frac 1 2}(|x|)}{\ln_{m+1}^{-\frac 1 2}(|x|)}
	&= 
		\frac{1}{4|x|^2}\prod_{k=1}^{m+1}\ln_k^{-2}(|x|)+\frac{2-d}{2|x|^2}\prod_{k=1}^{m+1}\ln_k^{-1}(|x|)\\
	  &\phantom{==} 
		+ \frac{1}{2|x|^2}\sum_{j=1}^{m+1}\prod_{s=1}^{m+1}\prod_{t=1}^j\ln_s^{-1}(|x|)\ln_t^{-1}(|x|)\,,\\
  2\frac{\partial_{r}\psi_{\ell,m}(x)}{\psi_{\ell,m}(x)}\frac{\partial_{r}\ln_{m+1}^{-\frac 1 2}(|x|)}{\ln_{m+1}^{-\frac 1 2}(|x|)}
	&= \left(
		  \frac{d}{|x|}+\frac{1}{|x|}\sum_{j=1}^{m}\prod_{k=1}^{j}\ln_k^{-1}(|x|)
		\right)
		\left(
		  \frac{1}{2|x|}\prod_{s=1}^{m+1}\ln_s^{-1}(|x|)
		\right)\,,\\
  W_{m+1}(x) = W_m(x)
    &+
      \frac{3}{4|x|^2}\prod_{k=1}^{m+1}\ln_k^{-2}(|x|)
        + \frac{1}{|x|^2}\prod_{k=1}^{m+1}\ln_k^{-1}(|x|)\\
      &
        +\frac{1}{|x|^2}\sum_{j=1}^{m}\prod_{s=1}^{m+1}\prod_{t=1}^j\ln_s^{-1}(|x|)\ln_t^{-1}(|x|)\,.
\end{split}
\end{equation*}
Since $W_0(x)= \frac{d(4-d)}{4|x|^2}$, this yields $W_m$ via induction.
\end{remark}
Now we come to the 
\begin{proof}[Proof of Theorem \ref{thm:intro absence} {\rm: }] 
  Using Lemma \ref{lem:Wm} and $\psi_{\ell,m}\in \calQ^{U_R}_{loc}(H)$ one sees that  
  \begin{equation} \label{eq:V below Wn}
  \begin{split}
	\la \varphi, H\psi_{\ell,m}\ra 
		&= \la \nabla\varphi,\nabla\psi_{\ell,m}\ra 
			+ \la \varphi, V\psi_{\ell,m}\ra 
		  \le \la \nabla\varphi,\nabla\psi_{\ell,m}\ra 
			+ \la \varphi, W_m\psi_{\ell,m}\ra \\
		& = \la \varphi, (-\Delta +W_m)\psi_{\ell,m}\ra 
		=0 
  \end{split}
  \end{equation}
  for all $0\le \varphi\in\calC^\infty_0(U_R)$
  as soon as $V\le W_m$ in $U_R$. 
  So $\psi_{\ell,m}$ is a zero energy subsolution of 
  $H$ in $U_R$ as soon as $V\le W_m$ in $U_R$. 
  Since $V$ is in the local Kato class we know from \cite{AizSim82, Simader90,Sim82} that 
  any eigenfunction of $H$ is continuous. Moreover, it is well--known that the ground state eigenfunction can be chosen to be strictly positive \cite{Far72, Goe77, ReeSim4}. 
  
  So if $\psi>0$ is a zero energy ground state of $H$ then, with 
  $c^1_R=\inf_{R\le |x|\le R+1}\psi(x)>0$ and 
  $c^2_R= \sup_{R\le |x|\le R+1} \psi_{\ell,m}(x) <\infty $, we can set 
  $C\coloneqq c^2_R/c^1_R$ to see that 
  $\psi_{\ell,m}(x)\le C\psi(x)$ for $R\le |x|\le R+1$. Moreover, 
  $\psi$ being a zero energy solution is also a zero energy supersolution, so Theorem \ref{thm:comp-agm} shows that 
  \begin{equation}
  	\psi_{\ell,m}(x)\le C\psi(x)
  \end{equation}
 for all $|x|>R$. In particular, $\psi$ cannot be 
 square integrable as soon as $V\le W_m$ on $U_R$ for some large enough $R>0$,  since $\psi_{\ell,m}$ is positive and not in $L^2(U_R)$. 

 Finally, setting $V_m(x)=  \frac{d(4-d)}{4|x|^2} 
		+\frac{1}{|x|^2}\sum_{j=1}^m\prod_{k=1}^j\ln_k^{-1}(x)$, we note that 
  $W_m(x)\ge V_m(x)$ for all large enough $|x|$. 
  This proves Theorem \ref{thm:intro absence}.
\end{proof}
\begin{remark}
  Of course, the proof of Theorem \ref{thm:intro absence} given above shows that if 
  \begin{equation}
    V(x)\le W_m(x) \quad \text{ for all } |x|>R	
  \end{equation}
  for large enough $R>0$ and some $m\in\N_0$, then the Schr\"odinger operator $H$ cannot have any zero energy ground state. 
\end{remark}

\section{Proof of the existence result}\label{sec:existence}
For the existence result we need to modify our comparison functions $\psi_{\ell,m}$
to make them barely square integrable near infinity. 
Given an arbitrary $\varepsilon>0$ we set  
\begin{equation}\label{eq:comparison functions upper bound}
\psi_{u,m,\epsilon}(x)\coloneqq\psi_{\ell,m}(x) \ln_m^{-\varepsilon/2}(|x|) 
\end{equation}
where $\psi_{\ell,m}$ is defined in \eqref{eq:comparison functions lower  bound}. 
For each $m\in\mathbb N_0$ and $\varepsilon>0$ we have 
$\psi_{u,m,\epsilon}\in \calC^\infty(\{|x|>e_m\})$  
and it is not hard to see that for any $\veps>0$ and any large enough $R>0$ 
the function $\psi_{u,m,\epsilon}$ is barely in $L^2(U_R)$. 
The potential $Y_{m,\epsilon}$ for which $(-\Delta +Y_{m,\epsilon})\psi_{u,m,\epsilon}=0$ 
in $U_R=\{x\in\R^d:\, |x|>R\}$ 
is given by  
\begin{lemma}\label{lem:Ym}
  For any $m\in\N_0$ we have $(-\Delta+Y_{m,\epsilon})\psi_{u,m,\epsilon}=0$ 
  in $U_{e_m}=\{x\in\R^d:\, |x|>e_m\}$, where the potential $Y_{m,\epsilon}$ is given by 
  \begin{equation}\label{eq:Ym}
  \begin{split}
  Y_{m,\epsilon}(x)
  	= W_m(x)  
  	& +\frac{\epsilon^2}{4|x|^2}\prod_{k=1}^{n}\ln_k^{-2}(|x|)
  			+\frac{\epsilon}{|x|^2}\prod_{k=1}^{n}\ln_k^{-1}(|x|)\\
  	  &
  	  	+\frac{\epsilon}{|x|^2}\sum_{j=1}^{n}\prod_{k=1}^{n}\prod_{m=1}^j\ln_k^{-1}(|x|)\ln_m^{-1}(|x|)\, ,
  \end{split}
  \end{equation}
  with $W_m$ given in \eqref{eq:Wm}. 
\end{lemma}
\begin{proof}
  As in the proof of Lemma~\ref{lem:Wm} we have to calculate $Y_{m,\epsilon}\coloneqq \frac{\Delta\psi_{u,m,\epsilon}}{\psi_{u,m,\epsilon}}$. We use 
  \eqref{eq:product rule} to see that 
  \begin{equation*}
	\begin{split}
	  Y_{m,\epsilon}(x)
	  	&=	\frac{\Delta\psi_{\ell,m}(x)}{\psi_{\ell,m}(x)} 
	  			+\frac{\Delta \ln_{m+1}^{-\frac\epsilon 2}(|x|)}{\ln_{m+1}^{-\frac\epsilon 2}(|x|)}
	  		    +2\frac{\partial_r \ln_{m+1}^{-\frac\epsilon 2}(|x|)}{\ln_{m+1}^{-\frac\epsilon 2}(|x|)}
	  		  	\frac{\partial_r \psi_{\ell,m}(|x|)}{\psi_{\ell,m}(|x|)}\\
		&=W_m(x)
			+\frac{\epsilon^2}{4|x|^2}\prod_{k=1}^{m}\ln_k^{-2}(|x|)
			+\frac{\epsilon}{|x|^2}
				\prod_{k=1}^{m}\ln_k^{-1}(|x|)\\
		&\phantom{=} +\frac{\epsilon}{|x|^2}\sum_{j=1}^{m}\prod_{s=1}^{m}\prod_{t=1}^j\ln_s^{-1}(|x|)\ln_t^{-1}(|x|) 
\end{split}
\end{equation*}
which is \eqref{eq:Ym}.    
\end{proof}
We want to show that ground states of Schr\"odinger operators $H$ with critical potentials $V$ exist using suitable eigenfunctions of $H_\lambda$. 
For this the following is convenient. 

\begin{lemma}\label{lem:existence bound state}
  Assume that the potential $V$ satisfies Assumption \ref{assumption} 
  and $W$ is a positive potential which is infinitesimally form small w.r.t.\ $-\Delta+V_+$.    
  Let  $(H_\lambda)_{\lambda\ge 0}$ be the family of 
  Schr\"odinger operators constructed in Remark \ref{rem:construction H-lambda}.  
  Moreover, assume that there exists a sequence $0<\lambda_n\to 0$ as $n\to\infty$ such that 
  the operators $H_n=H_{\lambda_n}$ have eigenvalues 
  $E_n=E_{\lambda_n}$ with corresponding normalized weak eigenfunctions 
  $\psi_n=\psi_{\lambda_n}$. If 
  \begin{SL}
  	\item[a{\rm)}] the sequence of eigenvalues $(E_n)_n$ of $H_n$ is bounded from above and 
  	\item[b{\rm)}] the sequence $\psi_{n}$ is a Cauchy sequence in $L^2$,
  \end{SL}
  then the sequence $\psi_{n}$ is Cauchy w.r.t.\ the quadratic form norm $\|\cdot\|_1$ given in 
  \eqref{eq:quadratic form norm}, hence its limit $\psi=\lim_{n\to\infty}\psi_{n}\in \calQ(H)$.
  Moreover  $E=\lim_{n\to\infty}E_{n}$ exists and $\psi$ is  
  a normalized weak eigenfunction of $H$ with eigenvalue $E$.
\end{lemma} 

\begin{remark}
  Lemma \ref{lem:regularity} shows even that $\psi\in \calD(H)$. Moreover, we do not need that 
  $V$ is in the local Kato--class, only that $V_-$ is relatively form small w.r.t.\ $-\Delta+V_+$.  
\end{remark}
\begin{remark}  
  We apply Lemma \ref{lem:existence bound state} when $\lambda_n$ converges monotonically to zero and  
  $E_n$ is a ground state of $H_n$, in which case 
  one can simplify the proof.  
  For example, if $E_n$ are ground state energies of $H_n$, then since as quadratic forms $H_{\lambda}\le H_{\lambda'}$ for all $0<\lambda'\le \lambda\le \lambda_0$, the limit 
  $\lim_{n\to\infty}E_n$ exists, by monotonicity. 
  However the result of Lemma \ref{lem:existence bound state} is needed when one considers not only 
  ground states, but also excited states which hit 
  the bottom of the essential spectrum. 
\end{remark} 

\begin{remark}\label{rem:E-delta bounded}
  Clearly any eigenvalue of $H_n$ is bounded from below uniformly in $n\in\N$ since $H_n\ge H_{\lambda_{max}}$ with $\lambda_{max}= \max_n\lambda_n$  
  as quadratic forms for all $n\in\N$. 
  In particular, all the eigenvalues $E_n$ are bounded uniformly in $n\in\N$ once they are bounded from above.  
  
  Moreover, if the essential spectrum of $H$ is not empty and $E_n$ is an eigenvalue of $H_n$ 
  below the essential spectrum of $H_n$, then $E_n$ is bounded from above, since as quadratic forms 
  $H_n\le H$ and by Persson's theorem \cite{LenSto19,Per60}  
  \begin{equation}
  \begin{split}
  	\inf&\sigma_{ess}(H_n)
  	  	= \lim_{R\to \infty} 
  	  	  \inf\big\{ 
  	  	  		\la \varphi, H_n\varphi \ra:\, \varphi\in \calQ(H), \|\varphi\|=1, \supp(\varphi)\subset\{|x|>R\}
  	  	  	\big\} \\
  	  &\le \lim_{R\to \infty} 
  	  		\inf\big\{ 
  	  	  		\la \varphi, H\varphi \ra:\, \varphi\in \calQ(H), \|\varphi\|=1, \supp(\varphi)\subset\{|x|>R\}
  	  	  	\big\} 
  	  		= \inf\sigma_{ess}(H)\, .
  \end{split}
  \end{equation}
  Thus  
  \begin{equation*}
  -\infty < \inf\sigma(H_{\lambda_{max}})\le \inf\sigma(H_n)\le E_n\le \inf\sigma_{ess}(H_n)\le \inf\sigma_{ess}(H)   	
  \end{equation*}
 for all $n\in\N$, which shows that $\sup_n|E_n|<\infty$  as soon as 
  $\sigma_{ess}(H)$ is not empty.  
\end{remark}
\begin{proof}[Proof of Lemma \ref{lem:existence bound state}{\rm:}] 
  Let $\psi_n$ be a normalized sequence of 
  eigenfunctions of $H_n$ with eigenvalue $E_n$ which is also a Cauchy sequence in $L^2$.  
  In particular,
  \begin{align}\label{eq:H-n eigenfunction 1}
  	\la\psi_n, H_n\psi_n \ra = E_n\la \psi_n, \psi_n\ra = E_n\, . 
  \end{align}
  Let $0<a_1<1$ and $b_1\ge 0$, respectively $a_2,b_2\ge0$,  
 such that  \eqref{eq:form small}, respectively \eqref{eq:form small} with $V_-$ replaced by $W$, holds. 
 Then 
 \begin{align*}
 	\la\psi_n, H_n\psi_n \ra 
 		&= \|\nabla\psi_n\|^2 + \|\sqrt{V_+}\psi_n\|^2 - \|\sqrt{V_-}\psi_n\|^2 -\lambda_n \|\sqrt{W}\psi_n\|^2 \\
 		&\ge \|\nabla\psi_n\|^2 + \|\sqrt{V_+}\psi_n\|^2 - (a_1+\lambda_n a_2)\|\nabla\psi_n\|^2 -(b_1+\lambda_n b_2)\|\psi_n\|^2 \\
 		&
 			\ge (1-a_1 - \lambda_n a_2)\|\psi_n\|_1^2 -(b_1 +\lambda_n b_2) \, ,
 \end{align*}
 since $\psi_n$ is normalized. We also used the 
 quadratic form norm $\|\cdot\|_1$  given 
 by \eqref{eq:quadratic form norm}. 
 Using \eqref{eq:H-n eigenfunction 1} this implies  
 \begin{equation}\label{eq:form norms bounded}
 	(1-a_1-\lambda_n a_2)\|\psi_n\|_1^2
 		\le b_1+\lambda_n b_2 + E_n \, ,
 \end{equation}
 which shows that we have 
  $\limsup_{n\to\infty}\|\psi_n\|_1<\infty$,  since 
 $a_1<1$, $\lambda_n\to 0$ for $n\to \infty$, and $E_n$ is bounded from above 
 uniformly in $n\in\N$.  Thus both the Sobolev norm 
 $\|\psi_n\|_{H^1}^2= \|\psi_n\|^2+\|\nabla\psi_n\|^2$ 
 and $\|\sqrt{V_+}\psi_n\|$ are bounded in $n$. 
 
 Now consider
 \begin{equation}\label{eq:convergence 1}
  \begin{split}
 	\la\varphi, H(\psi_n-\psi_m)\ra 
 		&= \la\varphi, H_n\psi_n\ra + \lambda_n\la \varphi, W \psi_n\ra 
 			- \la\varphi, H_m\psi_m\ra - \lambda_m\la \varphi, W \psi_m\ra \\
 		&= E_n\la\varphi, \psi_n\ra + \lambda_n\la \varphi, W \psi_n\ra
 			- E_m\la\varphi, \psi_m\ra - \lambda_m\la \varphi, W \psi_m\ra
  \end{split}
 \end{equation}
 for $\varphi\in\calQ(H)$. 
 The choice $\varphi=\psi_n-\psi_m$ and Cauchy--Schwarz yields  
 \begin{equation}\label{eq:convergence 2}
   \begin{split}
 	\la \psi_n &-\psi_m,  H(\psi_n-\psi_m)\ra \\
 		&\le |E_n|\|\varphi\|\|\psi_n\| + \lambda_n \|\sqrt{W}\varphi\|\|\sqrt{W} \psi_n\|  
 			+ |E_m|\|\varphi\|\|\psi_m\| + \lambda_m \|\sqrt{W}\varphi\|\|\sqrt{W}\psi_m\| \\
 		&\le  \big(|E_n|\|\psi_n\| + |E_m|\|\psi_m\|\big)\|\varphi\|  
 				+\frac{\lambda_n}{2} \big( \|\sqrt{W}\varphi\|^2+ \|\sqrt{W} \psi_n\|^2 \big)
 			\\
 			&\phantom{\le~~} 	
 			    +\frac{\lambda_m}{2} \big( \|\sqrt{W}\varphi\|^2+ \|\sqrt{W} \psi_m\|^2 \big) \\
 		&\le \big(|E_n|\|\psi_n\| + |E_m|\|\psi_m\|\big)\|\varphi\|  
 				+\frac{\lambda_n+\lambda_m}{2} \big( a_2\|\nabla\varphi\|^2+ b_2 \| \varphi\|^2 \big)
 			\\
 			&\phantom{\le~~} 	
 			    +\frac{\lambda_n}{2} \big( a_2 \|\nabla \psi_n\|^2+ b_2 \|\psi_n\|^2 \big)
+\frac{\lambda_m}{2} \big( a_2 \|\nabla \psi_m\|^2+ b_2 \|\psi_m\|^2 \big)
   \end{split}
 \end{equation}

 On the other hand, 
 \begin{align*}
 	\la \psi_n-\psi_m, H(\psi_n-\psi_m)\ra 
 		&= \|\nabla(\psi_n-\psi_m)\|^2 + \|\sqrt{V_+}(\psi_n-\psi_m)\|^2 
 				- \|\sqrt{V_-}(\psi_n-\psi_m)\|^2 \\
		&\ge (1-a_1)\|\nabla(\psi_n-\psi_m)\|^2 + \|\sqrt{V_+}(\psi_n-\psi_m)\|^2 
 				- b_1\|\psi_n-\psi_m\|^2 \, 
 \end{align*}
 and plugging this lower bound into \eqref{eq:convergence 2} and using that $\psi_n$ is normalized  we arrive at     	
 \begin{equation}\label{eq:convergence 3}
  \begin{split}
	 \big(1 -a_1-\frac{\lambda_n+\delta_m}{2}a_2\big)\|&\nabla(\psi_n-\psi_m)\|^2 + \|\sqrt{V_+}(\psi_n-\psi_m)\|^2 \\
 		&\le 
 			  \big(|E_n| + |E_m|+\frac{\lambda_n+\lambda_m}{2}b_2\big)\|\psi_n-\psi_m\| + b_1\|\psi_n-\psi_m\|^2 \\
 		  &\phantom{\le~~} 	
 			    +\frac{\lambda_n}{2} \big( a_2 \|\nabla \psi_n\|^2+ b_2 \big)
+\frac{\lambda_m}{2} \big( a_2 \|\nabla \psi_m\|^2+ b_2 \big)\, .
  \end{split}
 \end{equation}
 By assumption and Remark \ref{rem:E-delta bounded}, the sequence of eigenvalues $E_n$ is bounded and, 
 because of \eqref{eq:form norms bounded},  we also have that $\|\nabla\psi_n\|$ is bounded uniformly in $n\in\N$.  
 Since $\lambda_n\to 0$ and $\|\psi_n-\psi_m\|\to 0$ as $n,m\to\infty$,  \eqref{eq:convergence 3} implies
 \begin{equation*}
 	\limsup_{n,m\to\infty}\Big((1-a_1)\|\nabla(\psi_n-\psi_m)\|^2 + \|\sqrt{V_+}(\psi_n-\psi_m)\|^2 \Big) \le 0\, .
 \end{equation*}
 That is, the sequence of normalized weak eigenfunctions $\psi_n$ of $H_n$ is Cauchy in $\calQ(H)$ with respect to 
 the form norm $\|\cdot\|_1$ as soon as  it is Cauchy in $L^2$ and the sequence of eigenvalues $(E_n)_{n\in\N}$ is bounded.  
In particular, the limit $\psi=\lim_{n\to\infty}\psi_{n}$ exists in $\calQ(H)$. Thus  
$\|\nabla\psi\|= \lim_{n\to\infty}\|\nabla\psi_{n}\|$, $\|\sqrt{V_+}\psi\|= \lim_{n\to\infty}\|\sqrt{V_+}\psi_{n}\|$, and, since $W\ge 0$ is form bounded w.r.t. $-\Delta+V_+$ also $\|\sqrt{W}\psi\|= \lim_{n\to\infty}\|\sqrt{W}\psi_{n}\|$.
Hence $\sup_n \|\sqrt{W}\psi_n\| <\infty$.  
 
 Now assume additionally that $E_{n}$ converges to some $E$ as $n\to\infty$.    
 In this case, using that $\psi_n$ converges to $\psi$ in $\calQ(H)$ we get    
 \begin{equation}\label{eq:limit is weak eigenfunction}
 \begin{split}
 	\la \varphi, H\psi\ra &= \lim_{n\to\infty} \la \varphi, H\psi_n\ra 
 		= \lim_{n\to\infty}\big(\la \varphi, H_n\psi_n\ra + \lambda_n\la \sqrt{W}\varphi, \sqrt{W}\psi_n\ra \big)\\
 		&= \lim_{n\to\infty}\big(E_n\la \varphi, \psi_n\ra + \lambda_n \la \sqrt{W}\varphi, \sqrt{W}\psi_n\ra \big)
 			= E\la \varphi, \psi\ra
 \end{split}
 \end{equation}
 for all $\varphi\in\calQ(H)$ since $\lambda_n\to 0$ and 
 $\sup_n|\la \sqrt{W}\varphi, \sqrt{W}\psi_n\ra | \le \|\sqrt{W}\varphi\|\sup_n\|\sqrt{W}\psi_n\|<\infty$. 
 Thus we proved that the limit $\psi=\lim_{n\to\infty}\psi_n\in\calQ(H)$ 
 exists, $\|\psi\|=1$, and $\psi$ is a weak eigenfunction of $H$ with eigenvalue $E=\lim_{n\to\infty}E_n$ 
 under the \emph{additional assumption} that the limit  $E=\lim_{n\to\infty}E_n$ exists.  
 \smallskip
 
 Finally, it is easy to see that the sequence of eigenvalues $E_n$ must converge.  
 Assume that $E_n$ does not converge as $n\to\infty$.  
 Since $E_n$ is bounded in $n\in\N$, there exist two different limit points  $E_1\not = E_2$ of 
 $E_n$ corresponding to two subsequences 
 $E_{\sigma_1(n)}\to E_1$ and $E_{\sigma_2(n)}\to E_2$ where $\sigma_1, \sigma_2:\N\to\N$ are 
 strictly increasing functions.  
 
 Clearly $\psi=\lim_{n\to\infty}\psi_{\sigma_1(n)}= \lim_{n\to\infty}\psi_{\sigma_2(n)}$. 
 So \eqref{eq:limit is weak eigenfunction} shows that $\psi$ is a weak eigenfunction of $H$ corresponding 
 to the two different eigenvalues $E_1$ and $E_2$, which is impossible. 
 Hence the eigenvalues $E_n$ converges. This  
 finishes the proof of Lemma \ref{lem:existence bound state}.   
\end{proof}

\begin{lemma}\label{lem:local uniform boundedness of eigenfunctions of H-delta}
  Assume that the potentials $V$ and $W$ satisfy Assumption \ref{assumption}, except that the 
  relative bound of $W$ does not have to be less than one.  
  Let $(H_\lambda)_{0\le \lambda\le \lambda_0}$ be the family of 
  perturbed Schr\"odinger operators 
  constructed in Remark \ref{rem:construction H-lambda} for some small enough $0<\lambda_0$.  
  Moreover, assume that for some sequence $0<\lambda_n \le \lambda_0$ the operators $H_n=H_{\lambda_n}$ have eigenvalues 
  $E_n$ with corresponding weak eigenfunctions  $\psi_n$. 
  
  If $\|\psi_n\|=1$ for all $n\in\N$ and $\sup_{n}E_n<\infty$,    
  then the weak eigenfunctions $\psi_n$ are pointwise
  locally  bounded uniformly in $n\in\N$, i.e., 
  \begin{equation}\label{eq:locally uniformly bounded}
  	\sup_{n\in\N} \sup_{x\in S}|\psi_n(x)| <\infty 
  \end{equation}
  for any bounded set $S\subset\R^d$. 
\end{lemma}  
\begin{remark}
  Since eigenfunctions are continuous if the potential is locally in the Kato--class, $\psi_n(x)$ 
  makes sense for all $x\in\R^d$ and $n\in\N$.	
\end{remark}

\begin{proof}
 Note that $\psi_n$ is a zero energy weak eigenfunction of the Schr\"odinger operator 
 $\wti{H}_n$ with potential $\wti{V}_n$ given by 
 $ \wti{V}_n = V-\lambda_n W-E_n$.  If $V$ and $W$ are in the local Kato class, so is 
 $\wti{V}_n$. Hence for any $x\in\R^d$ the subsolution estimate
 \begin{equation}\label{eq:trudinger subsolution bound}
 	|\psi_n(x)| \le C_{x,n} \int_{|x-y|< 1} |\psi_n(y)|\, dy  
 \end{equation}
 holds, see \cite[Theorem C.1.2]{Sim82} and also \cite{AizSim82, Simader90}. 
 Moreover, the constants $C_{x,n}$ depend only on 
 \begin{equation*}
 	\|\ind_{B_{1}(x)}(\wti{V}_n)_-\|_{K^d}
 \end{equation*}
 with $(\wti{V}_n)_-$ being the negative part of $\wti{V}_n$ and the Kato norm $\|\cdot\|_{K^d}$ given by 
 \begin{equation}\label{eq:kato norm}
 	\|V\|_{K^d} 
 		\coloneqq
 			\sup_{x\in\R^d} \int_{|x-y|\le 1} \wti{g}_d(x-y)|V(y)|dy
 \end{equation}
 with $\wti{g}_d= g_d $ when $d\ge 3$, 
 $\wti{g}_2= 1+ g_2$, and $\wti{g}_1= 1$,  
 where $g_d$ is defined in \eqref{eq:g-d}.     
 Adding $1$ to $g_2$ is necessary since $g_2(x)=0$ when $|x|=1$. 

For any set $S\subset \R^d$ and any potential $V$ we have 
\begin{equation}
  \sup_{x\in S}\|\ind_{B_{1}(x)} V\|_{K^d} 
  	\le 
  		\|\sup_{x\in S}\ind_{B_{1}(x)} V\|_{K^d} 
  	= 
  		\|\ind_{S_1}V\|_{K^d} 
\end{equation}
where $S_1=\{y\in\R^d:\, \dist(y,S)< 1\}$. 
 
 \smallskip
 Now let $S\subset\R^d$ be bounded. Then $S_1$ is bounded and, since the Kato norm of a constant 
 function is finite and 
 $(\wti{V}_n)_-= (V-\lambda_n W -E_n)_{-}\le V_- +\lambda_n W_+ +(E_n)_+$, we have 
 \begin{equation*}
 	\sup_n\|\ind_{S_1}(\wti{V}_n)_-\|_{K^d} 
 		\le \big(\|\ind_{S_1}V_-\|_{K^d} + \sup_n \lambda_n \|\ind_{S_1}W_+\|_{K^d}+ \sup_n(E_n)_+\|1\|_{K^d}\big)
 		<\infty 
 \end{equation*}
 for any bounded set $S$, using that $\sup_nE_n<\infty$ and $\sup_n\lambda_n<\infty$, by assumption,  
 and $ \|\ind_{S_1}V_-\|_{K^d} <\infty$ and $ \|\ind_{S_1}W_+\|_{K^d} <\infty$, since $S_1$ is 
 bounded and $V$ and $W$ are  
 locally in the Kato class. 

 Thus for any bounded set $S\subset \R^d$ there exist a constant $C<\infty$ such that 
 \begin{equation}\label{eq:trudinger subsolution bound-2}
 	|\psi_n(x)| \le C \int_{|x-y|< 1} |\psi_n(y)|\, dy  
 \end{equation}
 for all $x\in S$ and $n\in\N$. Using the normalization $\|\psi_n\|=1$ we have 
 \begin{equation}
 	\int_{|x-y|< 1} |\psi_n(y)|\, dy 
 		\le 
 			|B^d_1|^{1/2}  \|\psi_n\|  
 		=  |B^d_1|^{1/2}
 \end{equation}
 for all $x\in S $ and  $n\in\N$. Hence \eqref{eq:locally uniformly bounded} 
 follows immediately from \eqref{eq:trudinger subsolution bound-2}.  
\end{proof}

The last result which we need is 
\begin{lemma}\label{lem:modulus is subsolution}
  Assume that $V\in L^1_{loc}(\R^d)$ and $V_-$ is form small w.r.t.\ $-\Delta+V_+$ and 
  $\psi$ is a real--valued weak eigenfunction of $H$ at energy $E$. 
  Then $|\psi|$ is a subsolution of $H$ at energy $E$.  
\end{lemma}
\begin{proof}
  If $\psi$ is real-valued eigenfunction of $H$ at energy $E$ then it is also a subsolution, hence 
  \cite[Lemma 2.9]{Agm85} shows that its positive part $\psi_+=\sup(\psi, 0)$ is a subsolution. 
  The same argument applied to $-\psi$, which is also a weak solution, shows that its negative part 
  $\psi_-=\sup(-\psi,0)$ is a subsolution. Hence $|\psi|= \psi_++\psi_-$ is a subsolution of $H$
  at energy $E$.   	
\end{proof}
\begin{remark}\label{rem:subsolutions complex valued}
  It is well--known that for the type of Schr\"odinger operator $H$ we consider here the eigenfunctions can be chosen to be real--valued. 
  Since $H$ is self--adjoint all eigenvalues are real. Moreover, $H$ commutes with complex conjugation, so 
  for any complex--valued eigenfunction $\psi$ of $H$ also the real and imaginary parts 
  $\re(\psi)=\frac{1}{2}(\psi+\ol{\psi})$ and $\im(\psi)=\frac{1}{2i}(\psi-\ol{\psi})$ are 
  eigenfunction of $H$ at energy $E$. 
  This is not true anymore if one considers Schr\"odinger operators with magnetic fields, 
  since they do not commute with complex conjugation, in general. 
\end{remark}

Now we are ready to give the  
\begin{proof}[Proof of Theorem \ref{thm:intro existence} {\rm:}]
%
  By assumption, the potential $V$ is critical.	
  Thus $\sigma(H)=\sigma_{ess}(H)=[0,\infty)$. 
  Moreover, for any non--trivial potential $W\ge 0$ 
  which is infinitesimally form small w.r.t.\ $-\Delta+V_+$ and has compact support the Schr\"odinger operators 
  $H_{\lambda}= H-\lambda W$, constructed in Remark \ref{rem:construction H-lambda},  have non-trivial discrete spectrum below zero. 
  That is,  $\sigma_{ess}(H_\lambda)= [0,\infty)$ and there exist eigenvalues $E_\lambda<0$ of $H_\lambda$ with associated normalized weak eigenfunctions 
  $\psi_\lambda$ for all $\lambda>0$. 
  We take any sequence  $(\lambda_n)_{n\in\N}$ which is  monotonically decreasing to zero and abbreviate $H_n=H_{\lambda_n}$, $E_n=E_{\lambda_n}$, and $\psi_n=\psi_{\lambda_n}$. 
  
  Recall that we also assume that the potential $V$ satisfies the lower bound 
  \begin{equation*}
  	V(x) \geq  \frac{d(4-d)}{4|x|^2}	+\frac{1}{|x|^2}\sum_{j=1}^m\prod_{k=1}^j\ln_k^{-1}(|x|)+\frac{2\epsilon}{|x|^2}\prod_{k=1}^m\ln_k^{-1}(|x|)  
  \end{equation*}
  for $|x|>R$, some $m\in\N_0$, $\epsilon>0$, and all large enough $R>0$. We replaced $\epsilon$ by $2\epsilon$ in \eqref{eq:intro existence}.  Increasing $R$, if necessary, it is easy to see that this implies  
  \begin{equation}\label{eq:lower bound V}
  	V(x)\ge Y_{m,\epsilon}(x) \text{ for all } |x|\ge R\, ,
  \end{equation}
  where the family of comparison functions $Y_{m,\epsilon}$ is defined in \eqref{eq:Ym}. 
   
  Since $W$ has compact support, we can also assume that  $R$ is so  large 
   that  its support $\supp(W)\subset B_R(x)$. Thus, with $U_R=\{|x|>R\}$ we have 
   $W\varphi = 0$ for all 
   $\varphi\in\calC^\infty_0(U_R)$.  
   Lemma \ref{lem:Ym} and \eqref{eq:lower bound V} imply        
   \begin{equation}
   \begin{split}
     \la \varphi, &(H_n-E_n) \psi_{u,m,\epsilon}\ra 
     	=      \la \varphi, (H-E_n) \psi_{u,m,\epsilon}\ra \\
     	&= \la \varphi, (-\Delta+Y_{m,\epsilon}-E_n) \psi_{u,m,\epsilon}\ra  
     		+ \la \varphi, (V- Y_{m,\epsilon})\psi_{u,m,\epsilon} \ra  
     	 \ge - E_n\la \varphi, \psi_{u,m,\epsilon} \ra
     		\ge 0
   \end{split}
   \end{equation}
   for all $0\le \varphi\in\calC^\infty_0(U_R)$. Here $\psi_{u,m,\epsilon}>0$ is 
   defined in \eqref{eq:comparison functions upper bound} and we used that $E_n\le 0$. 
   
   So for fixed $m\in\N$, large enough $R>0$, and small enough $\epsilon>0$ the  
   function $\psi_{u,m,\epsilon}$ is a supersolution of 
   $H_n$ at energy $E_n$  in $U_R$ for all $n\in\N$.  
   Moreover,  since $\|\psi_n\|=1$  we have 
   \begin{equation*}
   	 	  c^1_R\coloneqq \sup_{n\in\N} \sup_{R\le |x|\le R+1} |\psi_n(x)| <\infty 
   \end{equation*}
    by Lemma \ref{lem:local uniform boundedness of eigenfunctions of H-delta}.  
    Since $\psi_{u,m,\epsilon}>0$ is continuous away from zero, we also have  
  \begin{equation*}
  	c^2_R = \inf _{R\le |x| \le R+1 } \psi_{u,m,\epsilon}(x)>0\, 
  \end{equation*}
  and using  $C_R=c^1_R/c^2_R$ one gets $|\psi_n(x)|\le C_R \psi_{u,m,\epsilon}(x)$, hence also     
  \begin{equation}
  	\wti{\psi}_n(x)\coloneqq |\re(\psi_n(x))| + |\im(\psi_n(x))|\le \sqrt{2}|\psi_n(x)| \le \sqrt{2}C_R \psi_{u,m,\epsilon}(x)
  \end{equation}
  for all $R\le |x|\le R+1$ and all $n\in\N$.  
  Clearly, $|\psi_n|\le \wti{\psi}_n$. Since $\wti{\psi}_n$ is a nonnegative subsolution of 
  $H_n$  at energy $E_n$ by Lemma \ref{lem:modulus is subsolution} and Remark \ref{rem:subsolutions complex valued} we can  use  $w=\psi_{u,m,\epsilon}$ and $v= \widetilde{\psi}_n$ in  Theorem \ref{thm:comp-agm} to see that    
  \begin{equation}\label{eq:a priori upper bound}
  	|\psi_n(x)|\le \wti{\psi}_n(x) \le \sqrt{2}C_R \psi_{u,m,\epsilon}(x)\,\, \text{ for all } |x|\ge R
  \end{equation}
  uniformly in $n\in\N$. Since $\psi_{u,m,\epsilon}$ is square integrable at infinity 
  for any fixed $m\in\N$ and $\epsilon>0$, the bound \eqref{eq:a priori upper bound} yields tightness in 
  $x$-space, i.e., 
  \begin{equation}
  	\lim_{R\rightarrow\infty}\sup_{n\in\N}\int_{|x|>R}|\psi_n(x)|^2\textrm{d}x=0 \, .
  \end{equation}
  From \eqref{eq:form norms bounded} one gets $\sup_{n\in\N} \|\psi_n\|_{H^1}<\infty$. 
  In particular, we have  
  \begin{equation}
  	\lim_{L\rightarrow\infty}\sup_{n\in\N}\int_{|\eta|>L}|\widehat{\psi}_n(\eta)|^2\textrm{d}\eta=0 \, ,
  \end{equation}   
  which is tightness in momentum space. Here $\widehat{\psi}_n$ is the Fourier transform of $\psi_n$. 
  \smallskip
  
  Moreover, since $\psi_n$ is bounded in $H^1(\R^d)$, there exists 
  a subsequence which converges weakly in $H^1$ and $L^2$. 
  By a slight abuse of notation, we also write $\psi_n$ for this subsequence. 
  Let $\psi\in L^2(\R^d)$ be the weak limit of $\psi_n$. 
  Tightness and weak convergence then implies that $\psi_n$ converges to $\psi$ in $L^2$, 
  see e.g., \cite[Appendix A]{HunLee12}. 
  Hence $\|\psi\|= \lim_{n\to\infty}\|\psi_n\|=1$. 

 Lemma \ref{lem:existence bound state} shows that $E=\lim_{n\to\infty}E_n\le 0$ exists and that 
  $\psi$ is 
  a normalized weak eigenfunction of $H$ with eigenvalue $E$. Clearly, $E=0$ since $\sigma(H)=[0,\infty)$. 
  So zero is the ground state  eigenvalue of $H$ which is at the edge of the essential spectrum of $H$. 
  This finishes the proof of Theorem \ref{thm:intro existence}.  
\end{proof}

\begin{remark}\label{rem:existence under symmetry constraints}
  Note that we could have simplified some parts of the proof by using that ground states can be chosen to be strictly positive. 
  We intentionally avoided the use of strict positivity of ground state eigenfunctions. 
  This allows to use Theorem \ref{thm:intro existence} also for systems with symmetry restrictions, 
  or for the existence of higher eigenstates with energies above the ground state energy, provided 
  one suitably modifies the assumption of a virtual level for such systems. 
  These modifications are straightforward. 
\end{remark}

\begin{appendix}
\section{An example in search of a theorem}
\label{sec:appendix}
It is well--known that the zero potential is critical in dimensions one and two, 
see \cite{Sim76} and also \cite[Problems 1 and 2 in Chapter 45]{LanLif59-quantum-mechanics-non-relativistic}. This phenomenon can be explained by the 
non--integrability of 
$\eta\mapsto |\eta|^{-2}$ near $\eta=0$ in $\R^d$, see \cite{HoaHunRicVug17}.  
The Iorio-O'Carrol theorem \cite[Theorem XII.27]{ReeSim4} shows that shallow 
potential wells cannot create ground states  in dimension $d\ge 3$ and that 
the corresponding Schr\"odinger operators are even  unitarily equivalent to 
the free Laplacian. 

Of course, in order to construct zero energy resonances 
or zero energy ground states, 
one can take any Schr\"odinger operator $H$ which 
has essential spectrum 
$[0,\infty)$ and finitely many negative eigenvalues. 
Adding a suitable local positive perturbation then moves 
the ground state energy to zero, creating a zero 
energy resonance, or zero energy ground state, depending, 
for example, on which a priori bound from Theorem \ref{thm:intro absence} of Theorem \ref{thm:intro existence} holds.   

Specific examples of critical potentials in dimension one and two which are   
different from the zero potential seem to be  rare. 
In the following we construct a 
family of potentials $V_{\alpha,d}$ in any dimension which are critical 
for $\alpha\ge 0$, having a zero energy resonance when $0\le \alpha\le 1$ and 
a zero energy ground state when $\alpha>1$, and which are not critical 
when $\alpha<0$.   
To the best of our knowledge, our example is new.
\begin{remark}
  There are different definitions for a zero 
  energy resonance available in the literature. 
  One often calls $\psi$ a zero energy resonance if 
  it is a local positive eigenfunction of a  
  Schr\"odinger operator $H$ which is not square 
  integrable on $\R^d$ but its gradient $\nabla\psi$ 
  is square integrable. We will follow this convention, except that we also allow that the $L^2$--norm of 
  $\nabla\psi$ is logarithmically divergent at infinity.     
\end{remark}
For $\alpha\in \R$ and $d\in\N$ define the potential $V_{\alpha,d}$ on $\R^d$ by 
\begin{equation}\label{eq:V-alpha}
	V_{\alpha,d}(x)\coloneqq
		\frac{4\alpha^2 - (d-2)^2 }
		  {4\big(1+|x|^2\big)}
		  +\frac{1-(\alpha+d/2)^2}
		  {\big(1+|x|^2\big)^2}
\end{equation}
Clearly, $V_{\alpha,d}$ is bounded and goes to zero at infinity. Thus it is a Kato--class potential for all $d\ge 1$ and all $\alpha\ge 0$. In particular,  
$V_{\alpha,d}$ is both infinitesimally 
operator bounded, hence also infinitesimally 
form bounded,  w.r.t.\  $-\Delta$. Therefore    
the Schr\"odinger operator $H_{\alpha,d}= -\Delta+V_{\alpha,d}$ is a well 
defined self--adjoint operator on the domain $H^2(\R^d)$ 
with form domain $H^1(\R^d)$. 
\smallskip

The key to understanding why the potentials $V_{\alpha,d}$ are critical for all 
$\alpha\ge 0$ and $d\ge 1$, not critical for $\alpha<0$,  and switch from 
having zero energy resonances to having zero energy ground states at $\alpha=1$ 
is 

\begin{lemma}[Ground state representation of $H_{\alpha,d}$]\label{lem:ground state representation}
  Let $\alpha\in \R$, $d\ge 1$, and define 
  \begin{equation}\label{eq:psi-alpha-d}
  	\psi_{\alpha,d}(x)=(1+|x|^2)^{(2-d)/4 -\alpha/2}	
  \end{equation}
  for $x\in\R^d$ and the measure 
  \begin{equation}\label{eq:mu-alpha-d}
  	\mu_{\alpha,d}(B)= \int_B \psi_{\alpha,d}^2 \, dx \, 	
  \end{equation} 
  on the Borel sets $B$ in $\R^d$. 
  Then the map $U_{\alpha,d}:L^2(\R^d,d\mu_{\alpha,d})\to L^2(\R^d)$ given by 
  \begin{equation}\label{eq:ground state representation 0}
  	(U_{\alpha,d}\varphi)= \psi_{\alpha,d}\, \varphi   
  \end{equation}
  is unitary with 
  \begin{equation}\label{eq:ground state 1}
  	U_{\alpha,d}^{-1}(H^1(\R^d)) 
  		= 
  			\{\varphi\in L^2(\R^d,d\mu_{\alpha,d}):\, 
  				\nabla\varphi\in L^2(\R^d,d\mu_{\alpha,d})\}\, .
  \end{equation}
  Moreover, $U_{\alpha,d} H_{\alpha,d}U_{\alpha,d}^{-1} = -\Delta$ 
  in the sense that for all $\psi\in H^1(\R^d)$, the form domain 
  of $H_{\alpha,d}$, 
  \begin{equation}\label{eq:ground state representation}
  	\la \psi, H_{\alpha,d}\psi \ra 
  		= \la \nabla\psi, \nabla\psi \ra  - \la \psi, V_{\alpha,d}\psi\ra
  		=
  			\int_{\R^d} |\nabla\varphi|^2\, \psi_{\alpha,d}^2\, dx
  \end{equation}
  where $\varphi= U_{\alpha,d}^{-1}\psi$. 
%
\end{lemma}
\begin{remark}
  Lemma \ref{lem:ground state representation} shows that the 
  Schr\"odinger operator $H_{\alpha,d}$ is equivalent to the Dirichlet 
  form $ q(\varphi)= \la\nabla\varphi, \nabla\varphi \ra_{L^2(\R^d, d\mu_{\alpha,d})}$ 
  on the weighted $L^2$--space with measure 
  $d\mu_{\alpha,d}=(1+|x|^2)^{-(d-2)/2-\alpha}dx$. 
  Note that this measure is finite if and only if $\alpha>1$. 
\end{remark} 
We give the proof of the lemma at the end of the appendix. 
\begin{theorem}\label{thm: V-alpha-d}
  Let $d\in\N$, $\alpha\in\R$, and $H_{\alpha,d}= -\Delta+V_{\alpha,d}$ 
  be the self-adjoint Schr\"odinger operator with potential 
  $V_{\alpha,d}$ given by \eqref{eq:V-alpha}.  
  Then 
  \begin{SL}
  	\item\label{claim 1} $\sigma(H_{\alpha,d})=\sigma_{ess}(H_{\alpha,d})=[0,\infty)$. 
  	\item\label{claim 2} For all $\alpha\ge 0$ the potential $V_{\alpha,d}$ is critical, that is, the Schr\"odinger operator 
  			$H_{\alpha,d}$ has a virtual level. 
  	\item\label{claim 3} Zero is not an eigenvalue  of 
  		  $H_{\alpha,d}$ when $ 0\le \alpha\le 1$.  
  		  For  $\alpha>1$ zero is an eigenvalue. 
  		  The zero energy resonance for $0\le \alpha\le 1$, 
  		  respectively ground state for $\alpha>1$, is given by  
  		  \eqref{eq:psi-alpha-d}.
	\item\label{claim 4} For $\alpha<0$, the potential 
			$V_{\alpha,d}$ is subcritical, and zero is neither an eigenvalue nor a resonance. 
  \end{SL}
\end{theorem}
\begin{remark}
  Using the early result of Kato, \cite{Kat59}, see 
  also \cite{Agm70,Sim69}, the operator $H_{\alpha,d}$ 
  has no strictly positive embedded eigenvalues. 
  Since the potential $V_{\alpha,d}$ is short range, 
  the spectrum of $H_{\alpha,d}$ is even purely 
  absolutely continuous inside $(0,\infty)$, see \cite[Theorem 5.10]{CycFroKirSim87}.  
\end{remark}
\begin{proof}
  Using standard methods, \cite{Tes14}, one 
  sees that $\sigma_{ess}(H_{\alpha,d})=\sigma_{ess}(-\Delta)=[0,\infty)$ 
  since $V_{\alpha,d}$ is bounded and goes to zero at infinity.  
  Moreover, the ground state representation 
  \eqref{eq:ground state representation} implies  
  $\sigma(H_{\alpha,d})\subset [0,\infty)$. Hence 
  $\sigma(H_{a,d})= \sigma_{ess}(H_{a,d})=[0,\infty)$.    
  This proves claim \ref{claim 1}).   
  
  \smallskip
  
  Given $\varphi\in L^2(\R^d, d\mu_{\alpha,d})$ let $\psi\coloneqq U_{\alpha,d}\varphi$. 
  Taking $\varphi=1$ gives 
  $\psi= \psi_{\alpha,d}>0$ which is in $H^1(\R^d)$ 
  if and only if $\alpha>1$. 
  In this case \eqref{eq:ground state representation} shows that 
  $\psi_{\alpha,d}$ is the ground state of $H_{\alpha,d}$ 
  corresponding to the eigenvalue zero. 
  This proves the second claim in \ref{claim 3}). 
  Lemma \ref{lem:criticalV} also shows that the potential $V_{\alpha,d}$ is critical when $\alpha>1$. 
  
  In addition, note that the right hand side 
  of \eqref{eq:ground state representation} is strictly 
  positive unless  $\varphi$ is constant.  
  Hence  zero is not an eigenvalue  of $H_{\alpha,d}$ when 
  $0\le \alpha\le 1$ because $\psi_{\alpha,d}$ is not 
  square integrable in this case.   
  
  \smallskip
  
  When $0<\alpha\le 1$ we take any $\varphi\in \calC^\infty_0(\R)$ 
  with $\varphi(t)=1$ for $|t|\le 1$, $\varphi(t)=0$ for 
  $|t|\ge 2$, and define  
  \begin{align*}
  	\varphi_R(x) = \varphi(|x|/R)\, 
  \end{align*}
  for $R>0$. 
  Then $|\nabla\varphi_R(x)|= R^{-1}|\varphi'|(|x|/R)$. 
  Using $\psi_R= \psi_{\alpha,d}\, \varphi_R$ we get
  \begin{align*}
  	\la \psi_R, & H_{\alpha,d}\psi_R \ra 
  	  = 
  		R^{-2} \int |\varphi'(|x|/R)|^2 (1+|x|^2)^{(2-d)/2-\alpha}\, dx \\
  	 &\lesssim R^{-2}\int_R^{2R} (1+r^2)^{(2-d)/2-\alpha}\, r^{d-1} dr 
  	 	\sim   R^{-2} \int_R^{2R} (1+r^2)^{-\alpha}\, r dr 
	 \lesssim R^{-2\alpha}  
		\to 0
  \end{align*}
 for $R\to\infty$ and $\alpha>0$. Now let $W\ge 0$ have compact 
 support, be infinitesimally form bounded w.r.t.\ $-\Delta$,  and 
 $W>0$ on a set of positive Lebesque measure. Since 
 $\psi_R(x)\to (1+|x|^2)^{(2-d)/4-\alpha/2}$ as $R\to\infty$ 
 uniformly on compact sets we have  
 \begin{align*}
 	\lim_{R\to\infty} \la \psi_R, (H_{\alpha,d}-\lambda W)\psi_R \ra
 	= -\lambda \int W(x) (1+|x|^2)^{(2-d)/2-\alpha}\, dx <0
 \end{align*}
 for all $\lambda>0$. Thus 
 $\la \psi_R, (H_{\alpha,d}-\lambda W)\psi_R \ra <0$ for all 
 large enough $R>0$. Since 
 $\sigma_{ess}(H_{\alpha,d}-\lambda W) = [0,\infty)$, the 
 Rayleigh Ritz principle shows that $H_{\alpha,d}-\lambda W$ 
 has a negative eigenvalue for any $\lambda>0$. Thus the potential 
 $V_{\alpha,d}$ is critical. Clearly zero cannot be an eigenvalue nor 
 a resonance, since then the potential $V_{\alpha,d}$ would have 
 to be critical.  
 
 \smallskip
 
 To see that $V_{0,d}$ is critical one needs to modify the ansatz function. Let $\delta>0$ 
 and set 
 \begin{align*}
 	\varphi_\delta(x)
 	  \coloneqq 
 	    \left\{\begin{array}{ccc}
	      1 &\text{if}& |x|\le 1 \\
	      (1-\delta \ln|x|)_+ &\text{if}& |x|>1
        \end{array}\right.\, 
 \end{align*}
 and $\psi_\delta = U_{\alpha,d}\varphi_\delta$. A straightforward calculation shows 
   \begin{align*}
  	\la \psi_\delta, H_{0,d}\psi_\delta \ra 
  	  &= 
  		\delta^{2} \int_{1\le|x|\le e^{1/\delta}} (1+|x|^2)^{(2-d)/2}|x|^{-2}\, dx \\
  	 &\lesssim \delta^{2}\int_1^{e^{1/\delta}} (1+r^2)^{(2-d)/2}\, r^{d-3} dr 
  	 	\sim   \delta^{2} \int_1^{e^{1/\delta}} (1+r^2)^{-1}\, r dr \\
  	 &= \frac{\delta^2}{2} \ln(1+e^{1/\delta})  
		\to 0 \quad \text{as } \delta \to 0\, .
  \end{align*}
 Thus $\lim_{\delta\to 0}\la \psi_\delta, (H_{0,d}-\lambda W)\psi_\delta\ra= -\lambda \int W(x)(1+|x|^2)^{(2-d)/2}\, dx<0$. 
 As before this shows that $V_{0,d}$ is critical.   
 Moreover, even though $\psi_{\alpha,d}\not\in  L^2(\R^d)$ when $0\le \alpha\le 1$, its gradient 
 $\nabla\psi_{\alpha,d}$ is in $L^2(\R^d)$ when $0<\alpha\le 1$ and the $L^2$-norm of 
 $\nabla\psi_{0,d}$ is only logarithmically divergent.  
 Hence  $\psi_{\alpha,d}$ is a zero energy resonance  
 for $H_{\alpha,d}$ when $0\le \alpha\le 1$.  
 
 Finally, we look at $V_{-\alpha,d}$ for $\alpha>0$. 
 A simple calculation shows 
 \begin{align*}
 	V_{-\alpha,d}(x) 
 		= 
 			V_{\alpha,d}(x) + 2\alpha d (1+|x|^2)^{-2}\, . 
 \end{align*}  
 Thus with $W(x)= (1+|x|^2)^{-2}>0$ and $\lambda = 2\alpha d>0$ 
 we have 
 \begin{align*}
 	\la \psi, (H_{-\alpha,d}-\lambda W)\psi\ra
 		=
 			\la \psi, H_{\alpha,d}\psi\ra \ge 0
 \end{align*}
 for all $\psi\in H^1(\R^d)$, since $\sigma(H_{\alpha,d})=[0,\infty)$ 
 by part \ref{claim 1}).  
 Hence $V_{-\alpha,d}$ is subcritical for $\alpha>0$.  
\end{proof}

The family of potential  $V_{\alpha,d}$ has several interesting properties summarized in   
\begin{lemma}[Properties of $V_{\alpha,d}$]  
Let $\alpha\in\R$ and $d\in\N$. Then 
\begin{SL}
  \item\label{properties V-a} In dimensions $d=1,2$ the potential $V_{\alpha,d}$ is non--trivial if $\alpha\neq |d-2|/2$ and in dimension $d\ge 3$ it is non-trivial if $\alpha\neq (2-d)/2$.  
  \item\label{properties V-b} In dimension $d=1$ we have $V_{\alpha,1}>0$ for $\alpha \le -1/2$. If $-1/2<\alpha<1/2$ then $V_{\alpha,1}>0$ near zero and it has a negative tail, i.e., $V_{\alpha,1}(x)<0$ for large $|x|$. If $\alpha>1/2$, then $V_{\alpha,1}$ is negative near zero and it has a positive tail. 
  \item\label{properties V-c} In dimension $d=2$ we have  
  		$V_{\alpha,2}>0$, i.e., $V_{\alpha,2}$ is purely repulsive 
  		for all $\alpha<0$. For $\alpha>0$ the potential $V_{\alpha,2}$ 
  		is negative near zero and has a positive tail. 
  \item\label{properties V-d} In dimension $d\ge 3$ we have  
  		$V_{\alpha,d}>0$, i.e, the potential is repulsive, 
  		for $\alpha<(2-d)/2$. 
  		For $(2-d)/2 <\alpha\le (d-2)/2$ we have $V_{\alpha,d}<0$, 
  		i.e., the potential is attractive. 
  		If $\alpha>(d-2)/2$, then $V_{\alpha,d}$ is negative 
  		near zero and has a positive tail. 
  \item\label{properties V-e} For $d=1$ the potential $V_{\alpha,1}$ is integrable and 
  		\begin{equation*}
  			\int_{-\infty}^\infty V_{\alpha,1}(x)\, dx 
  			  =
  			    \frac{\pi}{2}(\alpha-1/2)^2 >0 \quad \text{for } \alpha\neq 1/2\, . 
  		\end{equation*}
  \item\label{properties V-f}  For large enough $R$ and all dimensions $d\ge 1$ the  potentials $V_{\alpha,d}$ satisfy  the bounds \eqref{eq:absence-m zero} for  $0\le \alpha<1$, respectively   \eqref{eq:intro absence m=1} for $\alpha=1$, while they  satisfy 
  the complementary bound \eqref{eq:existence-m zero} for $0<\epsilon<4(\alpha^2-1)$ when $\alpha>1$.  
\end{SL}
\end{lemma}
\begin{remark}
  Claims \ref{properties V-b}) and \ref{properties V-c}) above 
  are consistent with what is known about weakly coupled bound states 
  in low dimensions.  It is known that if 
  $V\in L^1(\R^d)$ and $\int_{\R^d} V\, dx\le 0$, where  $V$ is supposed to be non--trivial when $\int V \, dx=0$,  then the operator 
  $-\Delta+\lambda V$ always has a negative bound state, no matter how small the coupling parameter $\lambda>0$ is, when $d=1,2$. 
  See, for example, \cite{Sim76} where this is proved under some additional assumptions, or \cite{HoaHunRicVug17} for the full result. 
  In particular, this implies that critical potentials 
  in one and two dimensions have to change sign and, if 
  they are integrable, then $\int_{\R^d} V\, dx>0$ unless 
  $V$ is trivial.  
  
  In addition, claim \ref{properties V-d}) is consistent with our 
  non--existence Theorem \ref{thm:intro absence}. 
  Non--positive  potentials cannot have a zero energy ground state 
  in dimensions $d\le 4$. They need to have a 
  strong enough positive tail in order to be able to have zero 
  energy bound states. Moreover, claim \ref{properties V-f}) 
  together with the fact that the potential $V_{\alpha,d}$ 
  supports zero energy ground states if and only if $\alpha>1$, 
  see Theorem \ref{thm: V-alpha-d}, is consistent with our 
  Theorems \ref{thm:intro absence} 
  and \ref{thm:intro existence}.     
  
  It is illuminating to plot $V_{\alpha,d}(x)$ for $|x|=r$ to 
  explicitly see the behavior of  $V_{\alpha,d}$ for various 
  values of the parameters  $\alpha$ and $d$. 
%
\end{remark}
\begin{proof}
The first claim \ref{properties V-a}) is easy to check. To prove the rest, 
let $a_{\alpha,d}= \alpha^2-(d-2)^2/4$ and $b_{\alpha,d}= 1-(\alpha+d/2)^2$. Then 
  \begin{align*}
  	4V_{\alpha,d}(0) = 4(a_{\alpha,d}+b_{\alpha,d}) = -2d(d-2+2\alpha) >0  
  \end{align*}
  if and only if $\alpha<(2-d)/2$. 
  Moreover, unless $a_{\alpha,d}=0$, the sign of $V_{\alpha,d}(x)$ for large $|x|$ is determined by the sign of $a_{\alpha,d}$. Since 
  $a_{\alpha,d}>0$ if and only if $|\alpha|>|d-2|/2$, it is straightforward to deduce the claims \ref{properties V-b}), \ref{properties V-c}), and \ref{properties V-d}) from this.

  Clearly, $V_{\alpha,1}$ is integrable. Using  
  $\int_{-\infty}^\infty (1+x^2)^{-1}\, dx= \pi$ and 
  $\int_{-\infty}^\infty (1+x^2)^{-2}\, dx= \pi/2$,  
  claim \ref{properties V-e}) follows from a 
  simple calculation. 
 
 Since for large $|x|$ the second term in the definition of $V_{\alpha,d}$ is much 
 smaller than the first, the last claim \ref{properties V-f}) follows 
 from a straightforward computation. 
\end{proof}

It remains to give the proof of the ground state representation.

\begin{proof}[Proof of Lemma \ref{lem:ground state representation}]
  Let $\gamma\in \R $ and set 
  $\psi_\gamma(x)= (1+|x|^2)^{-\gamma/2}$ for 
  $x\in\R^d$, which is a regularized version of 
  $|x|^{-\gamma}$ used at the end of Section \ref{sec:DefAndResult}. When $\psi$ and $\varphi$ are related by  
  \begin{equation}\label{eq:relation psi varphi}
  	\psi= \psi_\gamma\, \varphi
  \end{equation}
  then $\psi\in L^2(\R^d)$ is clearly equivalent to 
  $\varphi\in L^2(\R^d, \psi_\gamma^2\, dx)$ and 
  the corresponding norms are the same. So the map 
  $U_\gamma:L^2(\R^d,\psi_\gamma^2\, dx) \to L^2(\R^d)$, $\varphi\mapsto \psi_\gamma\varphi$  
  preserves the corresponding norms. Its inverse is given by 
  $U_\gamma^{-1}\psi = U_{-\gamma}\psi = \psi_\gamma^{-1}\psi$ and from this 
  one easily checks that $U_\gamma$ is a unitary map 
  from the weighted space $L^2(\R^d, \psi_\gamma^2\, dx )$ 
  to $L^2(\R^d)$. This proves 
  \eqref{eq:ground state representation 0}.

  If $\psi\in L^2(\R^d)$ and $\varphi=U_\gamma^{-1}\psi\in L^2(\R^d, \psi_\gamma^2\, dx)$, then we have, 
  in the sense of distributions,   
  \begin{equation}\label{eq:nice 1}
  	\nabla\psi = \varphi\nabla\psi_\gamma + \psi_\gamma\nabla\varphi 
  		=
  			-\gamma \psi_\gamma (1+|x|^2)^{-1}x \varphi 
  			+ \psi_\gamma\nabla\varphi
  		= 
  			-\gamma (1+|x|^2)^{-1}x \psi 
  			+ \psi_\gamma\nabla\varphi\, 
  \end{equation}
  since $\psi_\gamma\in \calC^\infty(\R^d)$. 
  Clearly, $ (1+|x|^2)^{-1}x $ is bounded on 
  $\R^d$. 
  Therefore, if $\varphi\in L^2(\R^d, \psi_\gamma^2\, dx)$ and $\nabla\varphi\in L^2(\R^d, \psi_\gamma^2\, dx)$ then \eqref{eq:nice 1} shows that 
  $\nabla\psi\in L^2(\R^d)$.  
  Hence, if $\varphi$ and $\nabla\varphi$ are in 
  $L^2(\R^d, \psi_\gamma^2\, dx)$ then $\psi=U_\gamma\varphi$ 
  is in  $H^1(\R^d)$ . 
  
  Conversely, if $\psi\in H^1(\R^d)$, then, as distributions, 
  $
  	\nabla\varphi 
  		= 
  			\gamma (1+|x|^2)^{\gamma/2-1}x \psi 
  			+ (1+|x|^2)^{\gamma/2}\nabla \psi \, 
  $, which shows that 
  \begin{equation}\label{eq:nice 2}
  	\psi_\gamma\nabla\varphi 
  		=
  			(1+|x|^2)^{-1}x\psi + \nabla\psi \in L^2(\R^d)\, . 
  \end{equation}
  That is, if $\psi\in L^2(\R^d)$, then $\varphi=U_\gamma^{-1}\psi\in L^2(\R^d,\psi_\gamma^2\, dx)$ and if, in addition,  
  $\nabla\psi\in L^2(\R^d)$ then \eqref{eq:nice 2} 
  shows that  
  $\nabla\varphi\in L^2(\R^d,\psi_\gamma^2\, dx)$. 
  Altogether, this proves  
  \begin{equation*}
  	U_\gamma^{-1}(H^1(\R^d)) 
  		= 
  			\big\{
  				\varphi\in L^2(\R^d,\psi_\gamma^2\, dx):\, 
  				\nabla\varphi\in L^2(\R^d,\psi_\gamma^2\, dx)
  			\big\}\, 
  \end{equation*} 
  which is \eqref{eq:ground state 1}. 
  Moreover, $\calC^\infty_0(\R^d)$ is dense in $H^1(\R^d)$ and since $U_\gamma$ maps $\calC^\infty_0(\R^d)$ into itself it is also dense in $  	U_\gamma^{-1}(H^1(\R^d)) $. 
  So we only have to prove \eqref{eq:ground state representation} for  
  $\varphi\in \calC^\infty_0(\R^d)$. 
  
  Let $\gamma\in\R$ and  $\psi= \psi_\gamma\varphi$ with 
  $\varphi\in \calC^\infty_0(\R^d)$. 
  Then, as already noticed before,  
  \begin{equation*}
  	\nabla\psi(x) = -\gamma(1+|x|^2)^{-\gamma/2-1}x \varphi(x)
  		+ (1+|x|^2)^{-\gamma/2} \nabla\varphi(x)\, , 
  \end{equation*}
 hence 
 \begin{equation}\label{eq:nice 3}
 \begin{split}
   \la \nabla\psi,\nabla\psi \ra 
   	&= 
   		\la\nabla\varphi, (1+|x|^2)^{-\gamma}\nabla\varphi \ra 
   		-2\gamma\re(\la \nabla\varphi, (1+|x|^2)^{-\gamma-1}x\varphi \ra) \\
   	&\phantom{=~~}
   		+\gamma^2 \la \varphi, (1+|x|^2)^{-\gamma-2}|x|^2\varphi \ra \, .
 \end{split}	
 \end{equation}
An integration by parts shows 
\begin{align*}
  \re(\la &\nabla\varphi, (1+|x|^2)^{-\gamma-1}x\varphi \ra)
  	= -\re(\la \varphi, \nabla\cdot((1+|x|^2)^{-\gamma-1}x\varphi )  \ra  )\\
  	&= 2(\gamma+1) \la \varphi, (1+|x|^2)^{-\gamma-2}|x|^2\varphi \ra 
  		- d \la \varphi, (1+|x|^2)^{-\gamma-1}\varphi \ra 
  		- \re(\la \varphi, (1+|x|^2)^{-\gamma-1}x\nabla\varphi \ra )\, .
\end{align*} 
Noticing that $ \re(\la \varphi, (1+|x|^2)^{-\gamma-1}x\nabla\varphi \ra ) =   \re(\la \nabla\varphi, (1+|x|^2)^{-\gamma-1}x\varphi \ra) $ we get 
\begin{equation*}
    2\gamma \re(\la \nabla\varphi, (1+|x|^2)^{-\gamma-1}x\varphi \ra)
	=  2\gamma(\gamma+1) \la \varphi, (1+|x|^2)^{-\gamma-2}|x|^2\varphi \ra 
  		- d\gamma \la \varphi, (1+|x|^2)^{-\gamma-1}\varphi \ra \,
\end{equation*}
and plugging this into \eqref{eq:nice 3} we arrive at 
 \begin{equation*}
 \begin{split}
   \la \nabla\psi,\nabla\psi \ra 
   	&= 
   		\la\nabla\varphi, (1+|x|^2)^{-\gamma}\nabla\varphi \ra 
   		- 2\gamma(\gamma+1) \la \varphi, (1+|x|^2)^{-\gamma-2}|x|^2\varphi \ra \\
   	&\phantom{+~~~}	  		+ d\gamma \la \varphi, (1+|x|^2)^{-\gamma-1}\varphi \ra 
   		+\gamma^2 \la \varphi, (1+|x|^2)^{-\gamma-2}|x|^2\varphi \ra \, \\
    &=  
    	\la\nabla\varphi, (1+|x|^2)^{-\gamma}\nabla\varphi \ra 
      	- \gamma(\gamma+2-d) \la \varphi, (1+|x|^2)^{-\gamma-1}\varphi \ra  \\
   	&\phantom{+~~~}	
   		+\gamma(\gamma+2) \la \varphi, (1+|x|^2)^{-\gamma-2}\varphi \ra \\
   	&=  
    	\la\nabla\varphi, (1+|x|^2)^{-\gamma}\nabla\varphi \ra 
      	- \gamma(\gamma+2-d) \la \psi, (1+|x|^2)^{-1}\psi \ra  \\
   	&\phantom{+~~~}	
   		+\gamma(\gamma+2) \la \psi, (1+|x|^2)^{-2}\psi \ra\, .
 \end{split}	
 \end{equation*}
Choosing $\gamma= (d-2)/2+\alpha$ finishes the proof of Lemma \ref{lem:ground state representation}.
\end{proof}
\begin{remarks}
  The proof of Lemma \ref{lem:ground state representation} is clearly 
  inspired by the proof of Hardy's inequality on $L^2(\R^d)$ for 
  $d\ge 3$, where one considers $\psi(x)= |x|^{-\gamma/2}\varphi(x)$ 
  for $\varphi\in \calC^\infty_0(\R^d\setminus\{0\})$ and optimizes 
  in $\gamma>0$. One needs to restrict to $\varphi\in \calC^\infty_0(\R^d\setminus\{0\}) $ due to the singularity of 
  $|x|^{-\gamma/2}$ in zero. 
   Since $\calC^\infty_0(\R^d\setminus\{0\})$ is dense in $L^2(\R^d)$ only when $d\ge 3$, this leads to the well-known fact that 
   Hardy's inequality only holds in dimensions $d\ge 3$.  
\end{remarks}

\end{appendix}
\smallskip\noindent
\textbf{Acknowledgements:} 
Funded by the Deutsche Forschungsgemeinschaft (DFG, German Research Foundation) -- Project-ID 258734477 -- SFB 1173. 
This project also received funding from the 
European Research Council (ERC) under the 
European Union's Horizon 2020 research and 
innovation programme (grant agreement MDFT No.\ 725528).
Michal Jex also received financial support from the 
Ministry of Education, Youth and Sport of the Czech Republic under the Grant No.\ RVO 14000.  
Markus Lange was supported by NSERC of Canada and also acknowledges financial support from the European Research Council (ERC) under the European Union Horizon 2020 research and innovation programme (ERC StG MaMBoQ, grant agreement No. 802901).
It is also a pleasure to thank the CIRM, Luminy, for 
the REB  (research in residence) program, where part 
of this work was done.


\begin{thebibliography}{10}

\bibitem{Agm70}
S.~Agmon, \emph{Lower bounds for solutions of {S}chr\"{o}dinger equations}, 
 J.\  Analyse Math.\ \textbf{23} (1970), 1--25. \MR{276624} 
 \newblock \href {https://doi.org/10.1007/BF02795485}{\path{doi:10.1007/BF02795485}}
  ~ \hfill
  
\bibitem{Agm85}
S.~Agmon, \emph{Bounds on exponential decay of eigenfunctions of
  {S}chr{\"o}dinger operators}, Schr{\"o}dinger Operators (Berlin, Heidelberg)
  (Sandro Graffi, ed.), Springer Berlin Heidelberg, 1985, pp.~1--38. \MR{824986}, 
  \href {https://doi.org/10.1007/BFb0080331}{\path{doi:10.1007/BFb0080331}}
  \hfill
  
\bibitem{AizSim82}
M.~Aizenman and B.~Simon, \emph{Brownian motion and {H}arnack inequality for
  {S}chr{\"o}dinger operators}, 
\newblock Comm.\ Pure Appl.\ Math.\ \textbf{35} (1982), no.\ 22, 209--273.
\newblock \MR{644024}, \href {https://doi.org/10.1002/cpa.3160350206}
  {\path{doi:10.1002/cpa.3160350206}}.  \hfill
  
\bibitem{AvrHunHyn}
S.~Avramska-Lukarska, D.~Hundertmark, and H.~Kova\v{r}\'{\i}k, \emph{Absence of
  positive eigenvalues of magnetic {S}chr\"odinger operators}, 
  arXiv:2003.07294 (2020). 
\newblock \url{https://arxiv.org/abs/2003.07294} 
  \hfill
  
\bibitem{BarBit19}
S.~Barth and A.~Bitter, \emph{On the virtual level of two-body
  interactions and applications to three-body systems in higher dimensions},
  Journal of Mathematical Physics \textbf{60} (2019), no.~11, 113504.
  \newblock \MR{4032166},  \href {https://doi.org/10.1063/1.5120366}
  {\path{doi:10.1063/1.5120366}}. \hfill
  
\bibitem{BarBit20}
	S.~Barth and A.~Bitter.
	\newblock \emph{Decay rates of bound states at the spectral threshold of
  		multi-particle {S}chr\"{o}dinger operators.} 
	\newblock Doc.\ Math.\ \textbf{25} (2020), 721--735.
	\newblock \MR{4129671}, \href {https://doi.org/DOI: 10.25537/dm.2020v25.721-735}
  {\path{doi:DOI: 10.25537/dm.2020v25.721-735}}.
  \hfill 

\bibitem{BarBitVug19-a}
S.~Barth, A.~Bitter, and S.~Vugalter, \emph{Decay properties of
  zero-energy resonances of multi-particle {S}chr\"odinger operators and why the
  {E}fimov effect does not exist for systems of $ n \geq 4$ particles}, Preprint arXiv:1910.04139 (2019).
  \newblock \url{https://arxiv.org/abs/1910.04139}. \hfill
  
\bibitem{BarBitVug19-b}
	S.~Barth, A.~Bitter, and S.~Vugalter.
	\newblock \emph{The absence of the {E}fimov effect in systems of one- and
  	two-dimensional particles}. 
	\newblock J.\ Math.\ Phys.\ 62(12):Paper No. 123502, 46, 2021.
	\newblock \MR{4348075}, \href {https://doi.org/10.1063/5.0033524}
  {\path{doi:10.1063/5.0033524}}.
  \hfill 
\bibitem{BenYar90}
R.~D. Benguria and C.~Yarur, \emph{Sharp condition on the decay of the
  potential for the absence of a zero-energy ground state of the {S}chrodinger
  equation}, J.\ of Phys.\ A: Mathematical and General \textbf{23} (1990),
  no.~9, 1513--1518.
 \newblock \MR{1048781}, \url{http://stacks.iop.org/0305-4470/23/1513}.
  \hfill
  
\bibitem{bol85}
D.~Boll{\'e}, F.~Gesztesy, and W.~Schweiger, \emph{Scattering theory for
  long-range systems at threshold}, J.\ Math.\ Phys.\ \textbf{26}
  (1985), no.~7, 1661--1674. 
 \newblock \MR{793308}, \href {https://doi.org/10.1063/1.526963}
  {\path{doi:10.1063/1.526963}}.
  \hfill
  
\bibitem{CycFroKirSim87}
H.~L. Cycon, R.~G. Froese, W.~Kirsch, and B.~Simon, \emph{Schr\"odinger
  operators with application to quantum mechanics and global geometry}, study
  ed., Texts and Monographs in Physics, Springer-Verlag, Berlin, 1987.
  \MR{883643}
  \hfill
  
\bibitem{DeiHunSimVoc78}
P.~Deift, W.~Hunziker, B.~Simon, and E.~Vock, \emph{Pointwise bounds on
  eigenfunctions and wave packets in {$N$}-body quantum systems. {IV}}, 
 Comm.\ Math.\ Phys.\ \textbf{64} (1978/79), no.~1, 1--34. \MR{516993}
 \newblock URL: \url{http://projecteuclid.org/euclid.cmp/1103904619}.
  \hfill
  
\bibitem{DenKis07}
S.~A. Denisov and A.~Kiselev, \emph{Spectral properties of {S}chr\"{o}dinger
  operators with decaying potentials}, Spectral theory and mathematical
  physics: a {F}estschrift in honor of {B}arry {S}imon's 60th birthday, Proc.
  Sympos. Pure Math., vol.~76, Amer.\ Math.\ Soc., Providence, RI, 2007,
  pp.~565--589. \MR{2307748}
 \newblock \href {https://doi.org/10.1090/pspum/076.2/2307748}
  {\path{doi:10.1090/pspum/076.2/2307748}}.
  \hfill
  
\bibitem{DerSki09}
J.~Derezi\'{n}ski and E.~Skibsted, \emph{Quantum scattering at low energies},
  J.\ Funct.\ Anal.\ \textbf{257} (2009), no.~6, 1828--1920. \MR{2540993}
 \newblock \href {https://doi.org/10.1016/j.jfa.2009.05.026}
 {\path{doi:10.1016/j.jfa.2009.05.026}}.
  \hfill
  
\bibitem{Far72}
 W.~G. Faris.
 \newblock Quadratic forms and essential self-adjointness.
 \newblock {\em Helv. Phys. Acta}, 45:1074--1088, 1972/73.
 \newblock \MR{383964}, \url{http://doi.org/10.5169/seals-114428}.
 \hfill

\bibitem{FouSki04}
S.~Fournais and E.~Skibsted, \emph{Zero energy asymptotics of the resolvent for
  a class of slowly decaying potentials}, Math.\ Z.\ \textbf{248} (2004), no.~3,
  593--633. \MR{2097376}
 \newblock \href {https://doi.org/10.1007/s00209-004-0673-9}
  {\path{doi:10.1007/s00209-004-0673-9}}.
  \hfill
  
\bibitem{FraSim17}
R.~L. Frank and B.~Simon, \emph{Eigenvalue bounds for {S}chr\"{o}dinger
  operators with complex potentials. {II}}, J.\ Spectr.\ Theory \textbf{7}
  (2017), no.~3, 633--658. \MR{3713021}
 \newblock \href {https://doi.org/10.4171/JST/173} {\path{doi:10.4171/JST/173}}.
  \hfill
  
\bibitem{Goe77}
H.-W. Goelden.
\newblock \emph{On non-degeneracy of the ground state of {S}chr\"{o}dinger operators.}
\newblock Math.\ Z.\ \textbf{155} (1977), no.~3. 239--247.
\newblock \href {https://doi.org/10.1007/BF02028443}
  {\path{doi:10.1007/BF02028443}}. 
  \hfill 
  
\bibitem{GriGar07}
D.~K. Gridnev and M.~E. Garcia, \emph{Rigorous conditions for the existence of
  bound states at the threshold in the two-particle case}, J.\ Phys.\ A
  \textbf{40} (2007), no.~30, 9003--9016. \MR{2344533}
 \newblock \MR{609535}, \href {https://doi.org/10.1088/1751-8113/40/30/022}
  {\path{doi:10.1088/1751-8113/40/30/022}}.
  \hfill
  
\bibitem{HoaHunRicVug17}
	V.~Hoang, D.~Hundertmark, J.~Richter, and S.~Wugalter.
	\newblock \textit{Quantitative bounds versus existence of weakly coupled bound states	for {S}chr\"{o}dinger type operators}.
	\newblock  arXiv:arXiv:1610.09891, 40 pages, 2017.
	URL: \url{https://arxiv.org/abs/1610.09891}
	\hfill 
	
\bibitem{HofOstHofOstAhl78}
T~Hoffmann-Ostenhof, M.~Hoffmann-Ostenhof, and R.~Ahlrichs,
  \emph{"{S}chr\"odinger inequalities" and asymptotic behavior of many-electron
  densities}, Phys.\ Rev.\ A \textbf{18} (1978), 328--334.
 \newblock URL: \url{https://link.aps.org/doi/10.1103/PhysRevA.18.328}, \href
  {https://doi.org/10.1103/PhysRevA.18.328}
  {\path{doi:10.1103/PhysRevA.18.328}}.
  \hfill 
  
\bibitem{HunJexLan21-Helium}
D.~Hundertmark, M.~Jex, and M.~Lange.
\newblock \emph{Quantum systems at {T}he {B}rink: Helium--type systems}.
\newblock {\em arXiv:1908.04883} (2021), 62 pages.
 URL: \url{https://arxiv.org/abs/1908.04883}
\hfill

\bibitem{HunLee12}
D.~Hundertmark and Y.-R. Lee, \emph{On non-local variational problems with lack
  of compactness related to non-linear optics}, J.\ Nonlinear Sci.\  \textbf{22}
  (2012), no.~1, 1--38. \MR{2878650}
  \hfill
  
\bibitem{IonJer03}
A.~D. Ionescu and D.~Jerison, \emph{On the absence of positive eigenvalues of
  {S}chr\"{o}dinger operators with rough potentials}, Geom.\ Funct.\ Anal.\ 
  \textbf{13} (2003), no.~5, 1029--1081. \MR{2024415}
 \newblock \href {https://doi.org/10.1007/s00039-003-0439-2}
  {\path{doi:10.1007/s00039-003-0439-2}}.
  \hfill
  
\bibitem{JeKa79}
A.~Jensen and T.~Kato, \emph{Spectral properties of {S}chr{\"o}dinger operators
  and time-decay of the wave functions},  Duke Math.\ J.\ \textbf{46}
  (1979), no.~3, 583--611.
 \newblock URL: \url{http://projecteuclid.org/euclid.dmj/1077313577}.
  \hfill
  
\bibitem{JorWei73}
K.~J\"{o}rgens and J.~Weidmann, \emph{Spectral properties of {H}amiltonian
  operators}, Lecture Notes in Mathematics, Vol. 313, Springer-Verlag,
  Berlin-New York, 1973. \MR{0492941}
  \hfill
  
\bibitem{KaLo20}
K.~Kaleta and J.~L{\H{o}}rinczi, \emph{Zero-energy bound state decay for
  non-local {S}chr{\"o}dinger operators}, Commun.\ Math.\
  Phys.\ \textbf{374} (2020), no.~3, 2151--2191.
 \newblock \MR{4076095}, \href {https://doi.org/10.1007/s00220-019-03515-3}
  {\path{doi:10.1007/s00220-019-03515-3}}.
  \hfill
  
\bibitem{Kat59}
T.~Kato.
\newblock \textit{Growth properties of solutions of the reduced wave equation with a
  variable coefficient}.
\newblock Comm.\ Pure Appl.\ Math.\ \textbf{12} (1959), 403--425. 
\newblock \MR{108633}, 
\newblock \href {https://doi.org/10.1002/cpa.3160120302}
  {\path{doi:10.1002/cpa.3160120302}}.
  \hfill 

\bibitem{Ken89}
C.~E. Kenig, \emph{Restriction theorems, {C}arleman estimates, uniform
  {S}obolev inequalities and unique continuation}, Harmonic analysis and
  partial differential equations ({E}l {E}scorial, 1987), Lecture Notes in
  Math., vol. 1384, Springer, Berlin, 1989, pp.~69--90. \MR{1013816}
 \newblock \href {https://doi.org/10.1007/BFb0086794}
  {\path{doi:10.1007/BFb0086794}}.
  \hfill
  
\bibitem{KenNad00}
C.~E. Kenig and N.~Nadirashvili, \emph{A counterexample in unique
  continuation}, Math.\ Res.\ Lett.\ \textbf{7} (2000), no.~5-6, 625--630.
  \MR{1809288}
 \newblock \href {https://doi.org/10.4310/MRL.2000.v7.n5.a8}
  {\path{doi:10.4310/MRL.2000.v7.n5.a8}}.
  \hfill
  
\bibitem{KlSi-1}
M.~Klaus and Barry Simon, \emph{Coupling constant thresholds in nonrelativistic
  quantum mechanics. {I}. {S}hort-range two-body case}, Ann.\ Physics
  \textbf{130} (1980), no.~2, 251--281. \MR{610664}
 \newblock \href {https://doi.org/10.1016/0003-4916(80)90338-3}
  {\path{doi:10.1016/0003-4916(80)90338-3}}.
  \hfill 
  
\bibitem{Kno78-2}
  I.~Knowles, \emph{On the location of eigenvalues of second-order linear
  differential operators}, Proc.\ Roy.\ Soc.\ Edinburgh Sect.\ A \textbf{80}
  (1978), no.~1-2, 15--22. \MR{529565}
 \newblock \href {https://doi.org/10.1017/S030821050001009X}
  {\path{doi:10.1017/S030821050001009X}}. 
  \hfill
  
\bibitem{Kno78-1}
I.~Knowles, \emph{On the number of {$L^2$}-solutions of second order linear
  differential equations}, Proc.\ Roy.\ Soc.\ Edinburgh Sect.\ A \textbf{80}
  (1978), no.~1-2, 1--13. \MR{529564}
 \newblock \href {https://doi.org/10.1017/S0308210500010088}
  {\path{doi:10.1017/S0308210500010088}}.
  \hfill
  
\bibitem{KoTa02}
H.~Koch and D.~Tataru, \emph{Sharp counterexamples in unique continuation for
  second order elliptic equations}, J.\ Reine Angew.\ Math.\ \textbf{542} (2002),
  133--146. \MR{1880829}
 \newblock \href {https://doi.org/10.1515/crll.2002.003}
  {\path{doi:10.1515/crll.2002.003}}.
  \hfill
  
\bibitem{LanLif59-quantum-mechanics-non-relativistic}
	L.~D. Landau and E.~M. Lifshitz, 
	\newblock \emph{Quantum mechanics: non-relativistic theory. {C}ourse of
  {T}heoretical {P}hysics, {V}ol. 3}.
	\newblock Addison-Wesley Series in Advanced Physics. Pergamon Press, Ltd.,
  	London-Paris; for U.S.A. and Canada: Addison-Wesley Publishing Company, Inc.,
  	Reading, Mass;, 1958.
	\newblock Translated from the Russian by J. B. Sykes and J. S. Bell.
	\hfill

\bibitem{LeiSim81}
H.~Leinfelder and C.~G. Simader, \emph{Schr\"{o}dinger operators with singular
  magnetic vector potentials}, Math.\ Z.\ \textbf{176} (1981), no.~1, 1--19.
  \MR{606167}
 \newblock \href {https://doi.org/10.1007/BF01258900}
  {\path{doi:10.1007/BF01258900}}.
  \hfill
  
\bibitem{LenSto19}
D.~Lenz and P.~Stollmann, \emph{On the decomposition principle and a {P}ersson
  type theorem for general regular {D}irichlet forms}, J.\ Spectr.\ Theory
  \textbf{9} (2019), no.~3, 1089--1113. \MR{4003551}
 \newblock \href {https://doi.org/10.4171/JST/272} {\path{doi:10.4171/JST/272}}.
  \hfill
  
\bibitem{Lie81}
E.~H. Lieb, \emph{Thomas-{F}ermi and related theories of atoms and molecules},
  Rev.\ Modern Phys.\ \textbf{53} (1981), no.~4, 603--641. \MR{629207}
 \newblock \href {https://doi.org/10.1103/RevModPhys.53.603}
  {\path{doi:10.1103/RevModPhys.53.603}}.
  \hfill
  
\bibitem{Lun09}
D.~Lundholm.
\newblock \emph{Some spectral bounds for {S}chr\"{o}dinger operators with
  {H}ardy-type potentials}.
\newblock arXiv:0911.3386,  (2009).
\newblock URL: \url{https://arxiv.org/abs/0911.3386}.
\hfill

\bibitem{Lun-thesis10}
D.~Lundholm.
\newblock \emph{Zero-energy states in supersymmetric matrix models}.
\newblock PhD dissertation, KTH (2010).
\newblock URL: \url{http://urn.kb.se/resolve?urn=urn:nbn:se:kth:diva-12846}.
\hfill 

\bibitem{Nak94}
S.~Nakamura, \emph{Low energy asymptotics for {S}chr\"{o}dinger operators with
  slowly decreasing potentials}, Comm.\ Math.\ Phys.\ \textbf{161} (1994), no.~1,
  63--76. \MR{1266070} 
 \newblock URL: \url{http://projecteuclid.org/euclid.cmp/1104269792}.\hfill

\bibitem{New77}
R.~G. Newton, \emph{Nonlocal interactions; the generalized {L}evinson theorem
  and the structure of the spectrum}, J.\ Math.\ Phys.\ \textbf{18} (1977),
  no.~8, 1582--1588. \MR{446169} 
 \newblock \href {https://doi.org/10.1063/1.523466}
  {\path{doi:10.1063/1.523466}}.
 \hfill

\bibitem{PerCouMal86}
J.~F. Perez, F.~A.~B. Coutinho, and C.~P. Malta.
\newblock \emph{Logarithmic corrections to the uncertainty principle and infinitude
  of the number of bound states of {$N$}-particle systems}.
\newblock J.\ Math.\ Phys.\ \textbf{27} (1986), no./ 6, 1537--1540.
\newblock \href {https://doi.org/10.1063/1.527115}
  {\path{doi:10.1063/1.527115}}.
\hfill
  
\bibitem{Per60}
A.~Persson.
\newblock \emph{Bounds for the discrete part of the spectrum of a semi-bounded
  {S}chr\"{o}dinger operator}.
\newblock Math.\ Scand.\ \textbf{8} (1960), 143--153.
\newblock \href {https://doi.org/10.7146/math.scand.a-10602}
  {\path{doi:10.7146/math.scand.a-10602}}.
\hfill

\bibitem{Ram87}
A.~G. Ramm, \emph{Sufficient conditions for zero not to be an eigenvalue of the
  {S}chr\"{o}dinger operator}, J.\ Math.\ Phys.\ \textbf{28} (1987), no.~6,
  1341--1343. \MR{890004} 
 \newblock \href {https://doi.org/10.1063/1.527817}
  {\path{doi:10.1063/1.527817}}.
 \hfill

\bibitem{Ram88}
A.~G. Ramm, \emph{Conditions for zero not to be an eigenvalue of the
  {S}chr\"{o}dinger operator. ii}, J.\ Math.\ Phys.\ \textbf{29}
  (1988), no.~6, 1431--1432. 
 \newblock \href {http://arxiv.org/abs/https://doi.org/10.1063/1.527935}
  {\path{arXiv:https://doi.org/10.1063/1.527935}}, \href
  {https://doi.org/10.1063/1.527935} {\path{doi:10.1063/1.527935}}.
 \hfill

\bibitem{ReeSim4}
M.~Reed and B.~Simon, \emph{Methods of {M}odern {M}athematical {P}hysics {IV}.
  {A}nalysis of {O}perators}, Academic Press, New York-London, 1978.
  \MR{0493421} \hfill

\bibitem{Simader90}
C.~G. Simader, \emph{An elementary proof of {H}arnack's inequality for
  {S}chr\"{o}dinger operators and related topics}, Math.\ Z.\ \textbf{203}
  (1990), no.~1, 129--152. \MR{1030712} 
 \newblock \href {https://doi.org/10.1007/BF02570727}
  {\path{doi:10.1007/BF02570727}}. \hfill

\bibitem{Sim69}
B.~Simon, \emph{On positive eigenvalues of one-body {S}chr\"{o}dinger
  operators}, Comm.\ Pure Appl.\ Math.\ \textbf{22} (1969), 531--538. \MR{247300} 
 \newblock \href {https://doi.org/10.1002/cpa.3160220405}
  {\path{doi:10.1002/cpa.3160220405}}.
 \hfill

\bibitem{Sim76}
	B.~Simon.
	\newblock \emph{The bound state of weakly coupled {S}chr\"{o}dinger operators in one and two dimensions}.
	\newblock  Ann.\ Physics \textbf{97} (1976), no.~2, 279--288.
	\newblock \href {https://doi.org/10.1016/0003-4916(76)90038-5}
  		{\path{doi:10.1016/0003-4916(76)90038-5}}.
	\hfill 

\bibitem{Sim81}
B.~Simon, \emph{Large time behavior of the {$L^p$} norm of {S}chr{\"{o}}dinger
  semigroups}, J.\ Functional Analysis \textbf{40} (1981), no.~1, 66--83.
  \MR{607592} 
 \newblock \href {https://doi.org/10.1016/0022-1236(81)90073-2}
  {\path{doi:10.1016/0022-1236(81)90073-2}}.
 \hfill

\bibitem{Sim82}
B.~Simon, \emph{{Schr{\"o}dinger semigroups}}, Bulletin (New Series) of the
  American Mathematical Society \textbf{7} (1982), no.~3, 447--526.
 \newblock \url{https://doi.org/}, \href {https://doi.org/bams/1183549767}
  {\path{doi:bams/1183549767}}.
  \hfill

\bibitem{Tes14}
G.~Teschl, \emph{Mathematical methods in quantum mechanics}, second ed.,
  Graduate Studies in Mathematics, vol. 157, American Mathematical Society,
  Providence, RI, 2014, With applications to Schr\"{o}dinger operators.
  \MR{3243083} 
 \newblock \href {https://doi.org/10.1090/gsm/157} {\path{doi:10.1090/gsm/157}}. \hfill

\bibitem{VonNeuWig-paper-1929}
J.~von Neumann and E.~P. Wigner, \emph{{\"U}ber merkw{\"u}rdige diskrete
  {E}igenwerte}, pp.~291--293, Springer Berlin Heidelberg, Berlin, Heidelberg,
  1993. 
 \newblock \href {https://doi.org/10.1007/978-3-662-02781-3_19}
  {\path{doi:10.1007/978-3-662-02781-3_19}}.
 \hfill

\bibitem{Yaf82}
D.~R. Yafaev, \emph{The low energy scattering for slowly decreasing
  potentials}, Comm.\ Math.\ Phys.\ \textbf{85} (1982), no.~2, 177--196.
  \MR{675998}
 \newblock \url{http://projecteuclid.org/euclid.cmp/1103921410}.
  \hfill

\end{thebibliography}
%

\providecommand{\MR}{\relax\ifhmode\unskip\space\fi MR }
\providecommand{\MRhref}[2]{%
  \href{http://www.ams.org/mathscinet-getitem?mr=#1}{#2}
}
\providecommand{\href}[2]{#2}

\end{document}